%% file: ms.tex
\documentclass[11pt]{extarticle}
\pdfoutput=1
\input{macros}

\usepackage{enumitem}
\usepackage[para,online,flushleft]{threeparttable}
\newcommand{\bm}[1]{\mathbf{#1}}
\graphicspath{{./figure/}}

\newcommand{\Li}{L_{\textrm{instance}}}
\newcommand{\IN}{\mathrm{IN}}
\newcommand{\OUT}{\mathrm{OUT}}

\newcommand{\Q}{\mathcal{Q}}
\renewcommand{\E}{\mathcal{E}}
\newcommand{\V}{\mathcal{V}}
\newcommand{\R}{\mathcal{R}}
\newcommand{\T}{\mathcal{T}}
\newcommand{\A}{\mathcal{A}}

\newcommand{\dom}{\mathrm{dom}}

\newcommand{\eps}{\epsilon}
\renewcommand{\paragraph}[1]{\medskip\noindent{\bf {#1. }}}

\newcommand{\revise}[1]{{\color{black}{#1}}}

\begin{document}
	
\title{Instance and Output Optimal Parallel Algorithms for Acyclic Joins}

\author{
	{\sf Xiao Hu \ \ \ \ \ Ke Yi}\\
	Hong Kong University of Science and Technology\\
	\{xhuam, yike\}@cse.ust.hk
}

\date{}

\maketitle

\begin{abstract}
  Massively parallel join algorithms have received much attention in recent years, while most prior work has focused on worst-optimal algorithms. However, the worst-case optimality of these join algorithms relies on hard instances having very large output sizes, which rarely appear in practice.  A stronger notion of optimality is {\em output-optimal}, which requires an algorithm to be optimal within the class of all instances sharing the same input and output size.   An even stronger optimality is {\em instance-optimal}, i.e., the algorithm is optimal on every single instance, but this may not always be achievable. 

  In the traditional RAM model of computation, the classical Yannakakis algorithm is instance-optimal on any acyclic join.  But in the massively parallel computation (MPC) model, the situation becomes much more complicated.  We first show that for the class of r-hierarchical joins, instance-optimality can still be achieved in the MPC model.  Then, we give a new MPC algorithm for an arbitrary acyclic join with load $O ({\IN \over p} + {\sqrt{\IN \cdot \OUT} \over p})$, where $\IN,\OUT$ are the input and output sizes of the join, and $p$ is the number of servers in the MPC model.  This improves the MPC version of the Yannakakis algorithm by an $O (\sqrt{\OUT \over \IN} )$ factor.  Furthermore, we show that this is output-optimal when $\OUT = O(p \cdot \IN)$, for every acyclic but non-r-hierarchical join.  Finally, we give the first output-sensitive lower bound for the triangle join in the MPC model, showing that it is inherently more difficult than acyclic joins.
\end{abstract}

\section{Introduction}

A (natural) {\em join} is defined as a hypergraph $\Q = (\V, \E)$, where the vertices $\V=\{x_1,\dots,x_n\}$ model the {\em attributes} and the hyperedges $\E =\{e_1,\dots, e_m\} \subseteq 2^{\V}$ model the {\em relations}.  Let $\dom(x)$ be the {\em domain} of attribute $x \in \V$.  An {\em instance} of $\Q$ is a set of relations $\R = \{R(e): e \in \E\}$, where $R(e)$ is a set of {\em tuples}, where each tuple is an assignment that assigns a value from $\dom(x)$ to $x$ for every $x\in e$.  We use $\IN = \sum_{e \in \E} |R(e)|$ to denote the size of $\R$.  The {\em join results} of $\Q$ on $\R$, denoted as $\Q(\R)$, consist of all combinations of tuples, one from each $R(e)$, such that they share common values on their common attributes.  Let $\OUT= |\Q(\R)|$ be the output size.  We study the {\em data complexity} of join algorithms, i.e., we assume that the query size, namely $n$ and $m$, are constants.  In this paper, we focus on {\em acyclic joins}, i.e., when the hypergraph $\Q$ is acyclic (formal definition given later).  

\subsection{The model of computation}
The problem gets much more interesting in the parallel setting.  In this paper, we consider the {\em massively parallel computation} (MPC) model \cite{afrati2014gym,afrati11:_optim,beame13:_commun,beame14:_skew,ketsman17,koutris16:_worst,koutris11:_paral}, which has become the standard model of computation for studying massively parallel algorithms, especially for join algorithms.  

In the MPC model, data is initially distributed evenly over $p$ servers with each server holding $\IN/p$ tuples.  Computation proceeds in rounds. In each round, each server first sends messages to other servers, receives messages from other servers, and then does some local computation. The complexity of the algorithm is measured by the number of rounds and the {\em load}, denoted as $L$, which is the maximum message size received by any server in any round.  A {\em linear load} $L=O({\IN \over p})$ is the ideal case (since the initial load is already ${\IN \over p}$), while if $L=O(\IN)$, all problems can be solved trivially in one round by simply sending all data to one server.  Initial efforts were mostly spent \revise{on} what can be done in a single round of computation \cite{afrati11:_optim,koutris11:_paral,beame13:_commun,beame14:_skew,koutris16:_worst,koutris11:_paral}, but recently, \revise{more interest has} been given to multi-round (but still a constant) algorithms \cite{afrati2014gym,ketsman17,koutris16:_worst}, since new main memory based systems, such as Spark and Flink, have much lower overhead per round than previous generations like Hadoop.

The MPC model can be considered as a simplified version of the BSP model \cite{valiant90}, but it has enjoyed more popularity in recent years.  This is mostly because the BSP model takes too many measures into consideration, such as communication costs, local computation time, memory consumption, etc.  The MPC model unifies all these costs with one parameter $L$, which makes the model much 
\revise{simpler}.  Meanwhile, although $L$ is defined as the maximum incoming message size of a server, it is also closely related with the local computation time and memory consumption, which are both increasing functions of $L$.  Thus, $L$ serves as a good surrogate of these other cost measures.  This is also why the MPC model does not limit the outgoing message size of a server, which is less relevant to other costs.

All our algorithms work under the mild assumption $\IN \ge p^{1+\epsilon}$ where $\eps>0$ is any small constant.  This assumption clearly holds on any reasonable values of $\IN$ and $p$ in practice; theoretically, this is the minimum requirement for the model to be able to compute some almost trivial functions, like the ``or'' of $\IN$ bits, in $O(1)$ rounds.  Our lower bounds hold under $\IN \ge p^c$ for some constant $c$, which may depend on the particular lower bound construction.


We confine ourselves to {\em tuple-based} join algorithms, i.e., the tuples are atomic elements that must be processed and communicated in their entirety.  The only way to create a tuple is by making a copy, from either the original tuple or one of its copies.  We say that an MPC algorithm computes the join $\Q$ on instance $\R$ if the following is achieved: For any join result $(t_1,\dots, t_m) \in \Q(\R)$ where $t_i \in R(e_i)$, $i=1,\dots, m$, these $m$ tuples (or their copies) must all be present on the same server at some point.  Then the server will call a zero-cost function $emit(t_1,\dots, t_m)$ to report the join result.  Note that since we only consider constant-round algorithms, whether a server is allowed to keep the tuples it has received from previous rounds is irrelevant: if not, it can just keep sending all these tuples to itself over the rounds, increasing the load by a constant factor.  All known join algorithms in the MPC model are tuple-based and obey these requirements.  Our lower bounds are combinatorial in nature: we only count the number of tuples that must be communicated in order to emit all join results, while all other information can be communicated for free.  The upper bounds include all messages, with a tuple and an integer of $O(\log \IN)$ bits both counted as 1 unit of communication.


\subsection{Instance and output optimality}

In worst-case analysis, the entire space of instances 
\revise{is} divided into classes by the input size $\IN$, and the running time is measured on the worst instance in each class.  For many important computational problems, this is too coarse-grained and cannot accurately characterize the performance of the algorithm.  For the join problem, no algorithm can do better than $O(\IN^{1/\rho})$ time in the worst case, where $\rho$ is the fractional edge cover number of the hypergraph $\Q$ \cite{veldhuizen14,ngo2012worst}.  This bound drastically overestimates the running time on most typical instances.

A more refined approach is {\em parameterized analysis}, which further subdivides the instance space into smaller classes by introducing more parameters that supposedly better characterize the difficulty of each class. For the join problem, the output size $\OUT$ is \revise{a} commonly used parameter, and each class of instances share the same input and output size.  Let $\mathfrak{R}(\IN, \OUT)$ be the class of instances with input size $\IN$ and output size $\OUT$.  Then the load of an MPC algorithm $\A$ is thus a function of both $\IN$ and $\OUT$, defined as 
\[ L_\A(\IN, \OUT) = \max_{\R \in \mathfrak{R}(\IN, \OUT)} L_\A(\R),\]
where $L_\A(\R)$ denotes the load of $\A$ on $\R$.  Algorithm $\A$ is {\em output-optimal} if
\[
  L_\A(\IN, \OUT) = O(L_{\A'}(\IN, \OUT)),
\]
for every algorithm $\A'$. 

Further subdividing the instance space leads to more refined analyses. In extreme case when each class contains just one instance, we obtain instance-optimal algorithms.  More precisely, an algorithm $\A$ is {\em instance-optimal} if
\[
  L_\A(\R) = O(L_{\A'}(\R)),
\]
for every instance $\R$ and every algorithm $A'$.  Note that by definition, an instance-optimal algorithm must be output-optimal, and an output-optimal algorithm must be worst-case optimal, but the reserve direction may not be true.

In the traditional RAM model of computation, the classical Yannakakis algorithm \cite{yannakakis1981algorithms} can compute any acyclic join in time $O(\IN + \OUT)$, which is both output-optimal and instance-optimal, because on any instance $\R$, any algorithm has to at least spend $\Omega(\IN)$ time to read all the inputs\footnote{To formally prove this claim, one will have to be more careful with the family of algorithms under consideration.  In particular, if $\OUT=0$, then the algorithm may not have to do anything.  One possible approach is to ask the algorithm to produce a {\em certificate} in addition to the join results \cite{ngo14:_beyon}.  We will not digress to this direction since this paper is only concerned about MPC algorithms.} and $\Omega(\OUT)$ time to enumerate the outputs.  Thus, the two notions of optimality coincide (but both are stronger than worst-case optimality).  Fundamentally, this is because the difficulty of any instance $\R$ is precisely characterized by its input size and output size, and all instances in $\mathfrak{R}(\IN, \OUT)$ have exactly the same complexity $O(\IN + \OUT)$.

\subsection{Join algorithms in the MPC model}
The situation becomes much more interesting in the MPC model.  First, it has been observed that the Yannakakis algorithm can be easily implemented in the MPC model with a load of $O(\frac{\IN}{p} + \frac{\OUT}{p})$ \cite{afrati2014gym}\footnote{The bound stated in \cite{afrati2014gym} is actually $O ( {(\IN + \OUT)^2 \over p} )$, because they used a sub-optimal binary join algorithm as the subroutine.  Replacing it with the optimal binary join algorithm in \cite{beame14:_skew,hu17:_output} yields the claimed bound, as observed in \cite{koutris18:_algor}.}, but this is not optimal.  In particular, it is known that the binary join $R_1(A,B) \Join R_2(B,C)$ can be computed with load $O(\frac{\IN}{p} + \sqrt{\frac{\OUT}{p}})$ \cite{beame14:_skew,hu17:_output}.  This is optimal by the following simple lower bound argument: Each server can only produce $O(L^2)$ join results in a constant number of rounds with the load limited to $L$, so all the $p$ servers can produce at most $O(p \cdot L^2)$ join results.  Thus, producing $\OUT$ join results needs at least a load of $L = \Omega(\sqrt{\frac{\OUT}{p}})$.  Meanwhile, since $L\ge {\IN /p}$ by definition, the $O (\frac{\IN}{p} + \sqrt{\frac{\OUT}{p}} )$ bound is optimal.  Note that this argument can be applied on a per-instance basis, which means that the load complexity of any instance is still precisely captured by $\IN$ and $\OUT$, and $O (\frac{\IN}{p} + \sqrt{\frac{\OUT}{p}} )$ is both an instance-optimal and output-optimal bound.

However, when the join involves three relations, the situation becomes subtler, and we start to see a separation between the two notions of optimality, meaning that the load complexity of an instance may not depend only on $\IN$ and $\OUT$.  Let us start with the simplest 3-relation join $R_1(A) \Join R_2(B) \Join R_3(C)$, i.e., computing the Cartesian product of 3 sets of tuples.  Consider a particular class $\mathfrak{R}(\IN, \OUT)$ when $\OUT = \IN^2$.  Suppose the 3 relations have sizes $N_1, N_2, N_3$, respectively.  Then $\mathfrak{R}(\IN, \OUT)$ consists of all instances with $N_1+N_2+N_3 =\IN$ and $N_1N_2N_3=\OUT= \IN^2$.  Consider the following two instances: (1) $N_1=N_2=\Theta(\sqrt{\IN}), N_3=\Theta(\IN)$, applying the same argument above except that each server now can produce $O(L^3)$ join results, i.e., $p\cdot L^3 =\Omega(\OUT)$, we have $L= \Omega ( ({\OUT \over p} )^{1/3} )$; (2) if $N_1=1, N_2=N_3=\Theta(\IN)$, then the problem boils down to computing the Cartesian product of two sets, which has a lower bound of $L = \Omega ( (\frac{\OUT}{p} )^{1/2} )$.  The reason why instance (2) has a higher lower bound than instance (1) is that it has a higher skew, which causes more difficulty for the MPC model.  Note that this phenomenon does not exist in the RAM model, in which both instances (in fact all instances in $\mathfrak{R}(\IN, \OUT)$) have the same complexity of $O(\IN+\OUT)$.  Fundamentally, this is because the MPC model is all about {\em locality}: An MPC algorithm should strive to bring all related tuples to one server so as to produce as many join results as possible, while a higher skew reduces locality.

We can extend this argument to computing the Cartesian product of $m$ sets of sizes $N_1, \dots, N_m$.  Any algorithm computing the full Cartesian product obviously must also compute the Cartesian product of any subset of the $n$ sets, thus the load must be at least
\begin{equation}
  \label{eq:Cartesian}
L_{\mathrm{Cartesian}}(p,\R):=  \max_{S \subseteq\{1,\dots, m\}}  \left(\prod_{i\in S}N_i \over p\right)^{\frac{1}{|S|}}.
\end{equation}
It has been shown that the HyperCube algorithm \cite{afrati11:_optim} incurs a load of $L_{\mathrm{Cartesian}}(p,\R) \cdot \log^{O(1)} p$ on any instance $\R$ \cite{beame14:_skew}.  Thus, it can be considered as an instance-optimal algorithm for computing Cartesian products, with an optimality ratio of $\log^{O(1)} p$.  

The binary join and Cartesian products are the simplest joins.  Then the obvious question is, do instance-optimal algorithms exist for larger classes of joins?  If not, how about output-optimal algorithms?  These are the main questions we wish to address in this paper.

\subsection{Classification of acyclic joins}
Before describing our results, we first define some sub-classes of acyclic joins.

\medskip \noindent {\bf Acyclic joins \cite{beeri1983desirability}.}  We use the common notion of acyclicity, which is also known as $\alpha$-acyclicity.  A join $\Q = (\V,\E)$ is {\em acyclic} if there exists an undirected tree $\T$ whose nodes are in one-to-one correspondence with the edges in $\E$ such that for any vertex $v \in \V$, all nodes containing $v$ form a connected subtree. Such a tree $\T$ is called the {\em join tree} of $\Q$.  Note that the join tree may not be unique for a given $\Q$.

\medskip \noindent {\bf Hierarchical joins~\cite{dalvi2007efficient}.} A join $\Q = (\V,\E)$ is {\em hierarchical} if for every pair of vertices $x, y$, there is $\E_x \subseteq \E_y$, or $\E_y \subseteq \E_x$, or $\E_x \cap \E_y = \emptyset$, where $\E_x = \{e \in \E: x \in e\}$ is the set of hyperedges containing attribute $x$. Thus, all attributes can be organized into a forest, such that $x$ is a descendant of $y$ iff $\E_x \subseteq \E_y$.  Hierarchical joins have been enjoyed nice properties in probabilistic databases \cite{dalvi2007efficient,fink16:_dichot} and query answering under updates \cite{berkholz17:_answer}, but their role in the MPC model has not been studied so far.

\medskip \noindent {\bf r-hierarchical joins.}  We consider a slightly larger class of hierarchical joins. A {\em reduce} procedure on a hypergraph $(\V, \E)$ is to remove an edge $e \in \E$ if there exists another edge $e' \in \E$ such that $e \subseteq e'$. We can repeatedly apply the reduce procedure until no more edge can be reduced, and the resulting  hypergraph is said to be {\em reduced}.  A join is {\em r-hierarchical} if its reduced join hypergraph is hierarchical. A hierarchical join must be r-hierarchical, but not vice versa.  For example, the join $R_1(A) \Join R_2(A,B) \Join R_3(B)$ is r-hierarchical but not hierarchical. On the other hand, an r-hierarchical join must be acyclic.
 
 \medskip \noindent {\bf Tall-flat joins~\cite{koutris11:_paral}.} A join $\Q = (\V,\E)$ is {\em tall-flat} if one can order the attributes as $x_1, x_2, \cdots, x_h, y_1, \\y_2, \cdots, y_l$ such that (1) $\E_{x_1} \supseteq \E_{x_2} \supseteq \cdots \supseteq \E_{x_h}$; (2) $\E_{x_h} \supseteq \E_{y_j}$ for $j =1, 2, \cdots, l$; and (3) $|\E_{y_j}| = 1$ for $j =1, 2, \cdots, l$.  Obviously, a tall-flat join must be hierarchical.
  
\medskip The relationships of these joins are illustrated in Figure~\ref{fig:classification}.

 \begin{figure}[h]
 	\centering
 	\includegraphics[scale=0.9]{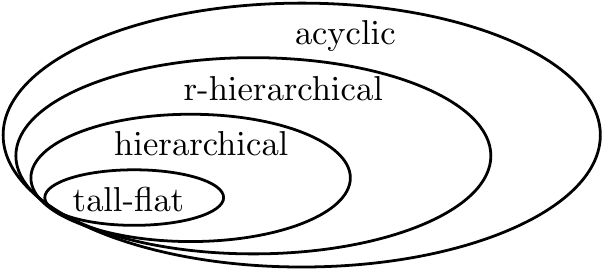}
 	\caption{Relationships of joins.}
 	\label{fig:classification}
 \end{figure}

\subsection{Our results}
        
This paper gives an almost complete characterization of acyclic joins with respect to instance-optimality and output-optimality in the MPC model.  Our results are summarized in Table~\ref{tab:summary}, and we explain them below in more detail.

{\tiny
	\begin{table*}
		\centering
		\begin{threeparttable}
			\begin{tabular}{c|c|c|c|c}
				\hline
				\multirow{2}{*}{Joins}& \multicolumn{2}{c|}{Instance-optimal\tnote{1}}  & \multicolumn{2}{c}{Output-optimal} \\  \cline{2-5}
				& one-round & multi-round & one-round & multi-round \\
				\hline
				tall-flat &  \multirow{2}{*}{$L_{\textrm{ins-opt}} \cdot \log^{O(1)} p$} & \multirow{5}{*}{$\Theta\left( L_{\textrm{ins-opt}} \right)$} & \multicolumn{2}{c}{\multirow{3}{*}{-}} \\ \cline{1-1}
				r-hierarchical  & \multirow{2}{*}{\cite{beame14:_skew}}& \\
				w/o dangling tuples && \\
				\cline{1-2} \cline{4-5}
				r-hierarchical & \multirow{2}{*}{$\omega\left( L_{\textrm{ins-opt}} \right)$} &  & \multirow{5}{*}{$\omega\left(\frac{\IN + \OUT}{p}\right)$} & \multirow{2}{*}{-} \\ 
				w/ dangling tuples  & & & \multirow{6}{*}{\cite{koutris11:_paral}} & \\
				\cline{1-3} \cline{5-5}
				\multirow{2}{*}{acyclic} & \multicolumn{2}{c|}{\multirow{3}{*}{$\omega\left( L_{\textrm{ins-opt}} \right)$}}& & $\Theta\left(\frac{\IN}{p} + \frac{\sqrt{\IN \cdot \OUT}}{p}\right)$\\  
				& \multicolumn{2}{c|}{}  && LB for $\OUT \le O( p \cdot \IN)$. \\
				\cline{1-1} \cline{5-5}
				triangle & \multicolumn{2}{c|}{} & &  $\tilde{\Omega}\left(\min\left\{\frac{\IN + \OUT}{p}, \frac{\IN}{p^{2/3}}\right\}\right)$\\ \hline
			\end{tabular}
		\end{threeparttable}
		\begin{enumerate}
			\item We use $L_{\textrm{ins-opt}}$ as short for $\frac{\IN}{p} + \Li(p, \R)$.
		\end{enumerate}
		\caption{Summary of results.}
		\label{tab:summary} 
	\end{table*}
}

\paragraph{Instance-optimality}
First, we extend the Cartesian product lower bound (\ref{eq:Cartesian}) to a general join $\Q = (\V, \E)$.  For any subset of relations $S\subseteq \E$, define
\[ \Q(\R, S) := \left(\Join_{e\in S} R(e) \right) \ltimes \Q(\R), \] i.e., the join results of relations in $S$ that are part of a full join result.  Clearly, any algorithm computing $\Q(\R)$ must implicitly also compute $\Q(\R, S)$ for every $S$.  Because each join result in $\Q(\R,S)$ consists of $|S|$ tuples, one from each relation in $S$, a server can emit at most $O(L^{|S|})$ join results of $\Q(\R,S)$, so we must have $p\cdot L^{|S|} =\Omega(|\Q(\R, S|)$.  Thus, we obtain the following per-instance lower bound on the load:
\begin{equation}
  \label{eq:instance}
	\Li(p,\R) := \max_{S \subseteq \E}\ \left(\frac{|\Q(\R, S)|}{p}\right)^{\frac{1}{|S|}}.
\end{equation}

The BinHC algorithm \cite{beame14:_skew} is \revise{a} generalization of the HyperCube algorithm to general joins.  The load of the BinHC algorithm is parameterized by the degrees of all subsets of attribute values (more detail given in Section~\ref{sec:hierarchical}).  Beame et al.~\cite{beame14:_skew} show that BinHC is optimal (up to polylog factors) within the class of instances sharing the same degrees, among all one-round MPC algorithms.  In this paper, we strengthen this result by giving a new analysis of the BinHC algorithm, showing that it is actually instance-optimal (up to polylog factors) for (1) all tall-flat joins, and (2) all r-hierarchical joins provided that the instance does not contain {\em dangling tuples} (a dangling tuple is one that does not appear in the join results). Furthermore, because the per-instance lower bound (\ref{eq:instance}) also holds for multi-round algorithms, these instance-optimality results extend to multi-round algorithms as well.  For r-hierarchical joins with dangling tuples, one-round algorithms cannot achieve $O({\IN \over p}+ \Li(p, \R))$ load, but we can remove the dangling tuples in $O(1)$ rounds with $O({\IN \over p})$ load \cite{yannakakis1981algorithms}, and then run then BinHC algorithm.  This gives a multi-round, $({\IN \over p} + \Li(p, \R)) \log^{O(1)} p$-load algorithm, where the $O(1)$ exponent depends on the query size, and is at least $m$, the number of relations. Then we give a new multi-round algorithm for r-hierarchical joins with load $O({\IN \over p} + \Li(p, \R))$, i.e., improving the instance-optimality ratio from $\log^{O(1)} p$ to $O(1)$.

The instance-optimal load $O({\IN \over p} + \Li(p, \R))$ is not achievable beyond r-hierarchical joins\footnote{But instance-optimal algorithms are still possible, if some higher per-instance lower bound can be derived.}.  More precisely, we show that for every acyclic join that is not r-hierarchical, there is an instance $\R$ with $\Li(p,\R)=O({\IN \over p})$ but any multi-round algorithm must incur a load of\footnote{The $\tilde{O}$ and $\tilde{\Omega}$ notation suppresses polylog factors.} $\tilde{\Omega} ({\IN \over p^{1/2}} )$ on $\R$.  This is actually a corollary following our output-sensitive lower bound, which is described next.

\paragraph{Output-optimality}
One-round algorithms have severe limitations with respect to $\OUT$: As shown in \cite{koutris11:_paral}, any non-tall-flat joins must incur load $\omega({\IN \over p} + {\OUT \over p})$ if only one round is allowed.  On the other hand, as mentioned, the classical Yannakakis algorithm is a multi-round MPC algorithm that works for all acyclic joins and has a load of $O(\frac{\IN}{p} + \frac{\OUT}{p})$ \cite{afrati2014gym,koutris18:_algor}.  Thus, our focus will be on multi-round algorithms and see if this result can be improved.  An instance-optimal algorithm must also be output-optimal, so we have automatically obtained output-optimal algorithms for r-hierarchical joins.  In fact, 
we show that $\Li(p, \R) = O ({\IN \over p} + \sqrt{\OUT \over p} )$ for all r-hierarchical joins, so this is already an asymptotic improvement over the Yannakakis algorithm.  But the more important question is, how about acyclic joins that are not r-hierarchical?

Our main output-optimal result is a new MPC algorithm for acyclic joins achieving a load of $O(\frac{\IN}{p} + \frac{\sqrt{\IN \cdot \OUT}}{p})$, which is an $O(\sqrt{\OUT \over \IN})$-factor improvement from the Yannakakis algorithm.  Interestingly enough, we observe that while the join order does not change the running time of the Yannakakis algorithm by more than a constant factor in the RAM model, it does have asymptotic consequences in the MPC model.  However, there are instances on which no join order is good, in which case we recursively decompose the join into multiple parts, and choose a good join order for each part.  The number of parts is exponential in the query size but constant in terms of data size.  To achieve this result, we first give a simple algorithm on the line-3 join $R_1(A,B) \Join R_2(B,C) \Join R_3(C,D)$ (Section~\ref{sec:line3}), and then extend it to arbitrary acyclic joins (Section~\ref{sec:full-acyclic}).

We also give a matching lower bound (up to a log factor), thereby establishing the output-optimality of the algorithm.  However, the lower bound only holds when $\OUT = O(p \cdot \IN)$.  This restriction on $\OUT$ is actually inherent, because the $O (\frac{\IN}{p} + \frac{\sqrt{\IN \cdot \OUT}}{p} )$ bound cannot be optimal for all values of $\OUT$.  When $\OUT$ is large enough, a worst-case optimal algorithm will take over. For example, on the line-3 join, the worst-case optimal algorithm, which has load $O ({\IN \over \sqrt{p}} )$ \cite{koutris16:_worst,hu2016towards}, becomes better when $\OUT > p \cdot \IN$.  Our lower bound actually indicates that the $O ({\IN \over \sqrt{p}} )$ bound is output-optimal (though it does not depend on $\OUT$) for all $\OUT > p \cdot \IN$.  Thus, we now have a complete understanding of the line-3 join with respect to output-optimality.  For more complicated joins, their worst-case optimal algorithms have a higher load, and the output-optimality for $\OUT$ values in the middle is still unclear. 

Next, we extend these results to {\em join-aggregate} (including {\em join-project}) queries that are {\em free-connex} (formal definition given in Section~\ref{sec:join-aggregate}).  More precisely, we give an MPC algorithm with linear load that removes all the non-output attributes of the query, converting it into an acyclic join.  Then we apply our instance-optimal or output-optimal algorithm on the resulting acyclic join.  

Finally in Section~\ref{sec:triangle}, we turn to the triangle join $R_1(B,C)\Join R_2(A,C) \Join R_3(A,B)$, which is the simplest cyclic join, and give the first output-sensitive lower bound $\tilde{\Omega}(\min\{\frac{\IN}{p} + \frac{\OUT}{p}, \\\frac{\IN}{p^{2/3}}\})$ in the MPC model. Previously, only a worst-case bound of $\Omega(\frac{\IN}{p^{2/3}})$ is known \cite{koutris16:_worst,pagh14} and that construction uses an instance with the maximum possible output size $\OUT=\IN^{3/2}$.  Note that the second term in the lower bound is smaller as long as $\OUT = \Omega(\IN \cdot p^{1/3})$, which means that under this parameter range, the $\tilde{O}(\frac{\IN}{p^{2/3}})$-load algorithm~\cite{koutris16:_worst} is not only worst-case optimal but also output-optimal.  For $\OUT = o(\IN \cdot p^{1/3})$, the lower bound becomes $\tilde{\Omega}({\IN \over p} + {\OUT \over p})$ while we do not have a matching upper bound yet (some explanation on why this is difficult is given below).  But at least, this shows a separation from acyclic joins, i.e., cyclic joins are harder than acyclic ones by at least a factor of $\tilde{\Omega}(\sqrt{\OUT \over \IN})$.

\subsection{Other related results}
\label{sec:other-relat-results}

Most existing work on join algorithms in the MPC model has focused on the worst case.  Here, the goal is to achieve a load of $O ({\IN \over p^{1/\rho}} )$, where $\rho$ is the fractional edge cover number of the hypergraph $\Q$.  So far, this bound has been achieved on Berge-acyclic joins\footnote{A sub-class of $\alpha$-acyclic joins.}  \cite{hu2016towards}, joins where each relation has two attributes (i.e., $\Q$ is an ordinary graph) \cite{ketsman17}, and LW joins \cite{koutris16:_worst}\footnote{The LW join algorithm presented in  \cite{koutris16:_worst} has a mistake, but it can be fixed, although non-trivially.}.  Whether this bound can be achieved for arbitrary joins, or even just $\alpha$-acyclic joins, is still open.  Assuming this is achievable, our output-sensitive algorithm is still better when $\OUT = O(p^{2-2/\rho}\cdot \IN)$. 

Joglekar et al.~\cite{manas16:_its} described a multi-round MPC algorithm for arbitrary joins, whose load complexity depends on $\IN, \OUT$, as well as the degrees of the values.  However, the load of their algorithm is at least $\Omega (\frac{\IN}{p} + \frac{\OUT}{p} )$, i.e., no better than the Yannakakis algorithm on acyclic joins. 

In the RAM model, output-sensitive join algorithms have been extensively studied.  The running time of most algorithms is in form of $O(\IN^w + \OUT)$, where $w$ is certain notion of {\em width} of the hypergraph $\Q$ \cite{gottlob05:_hyper,grohe14:_const,marx13:_tract,khamis17:_what}. However, it is not clear if this is optimal.  Even for the triangle join, it is not known what the output-optimal bound is.  For the triangle join, any notion of width has $w\ge 1.5$, thus these algorithms are no better than the worst-case optimal algorithm, which has running time $O(\IN^{1.5})$.  Recently, an improved triangle algorithm has been developed with a running time of $O(\IN^{1.408} + \IN^{1.222}\OUT^{0.186})$ \cite{orklund14:_listin}, which is better than the worst-case optimal algorithm when $\OUT < \IN^{1.495}$.  On the lower bound side, it is known that when $\OUT \ge \IN$, at least $\Omega(\IN^{4/3-o(1)})$ time is needed, assuming the 3SUM conjecture \cite{scu10:_towar}.  Thus, output-optimal algorithms for cyclic joins still remain a wide open problem.

 \section{MPC Primitives}
 \label{sec:primitive}

 Assume $\IN>p^{1+\epsilon}$ where $\epsilon >0$ is any small constant. We first introduce the following primitives in the MPC model, all of which can be computed with linear load $O (\frac{\IN}{p} )$ in $O(1)$ rounds.


 \medskip \noindent {\bf Multi-numbering~\cite{hu17:_output}:} Given $\IN$ (key, value) pairs, for each key, assigns consecutive numbers $1, 2, 3, \dots$ to all the pairs with the same key.

 \medskip \noindent {\bf Sum-by-key~\cite{hu17:_output}:} Given $\IN$ (key, value) pairs, compute the sum of values for each key, where the sum is defined by any associative operator.

 \medskip \noindent {\bf Multi-search~\cite{hu17:_output}:}  Given $N_1$ elements $x_1, x_2, \cdots, x_{N_1}$ as set $X$ and $N_2$ elements $y_1, y_2, \cdots, y_{N_2}$ as set $Y$, where all elements are drawn from an ordered domain.  Set $\IN =N_1+N_2$.  For each $x_i$, find its predecessor in $Y$, \revise{i.e., the largest element in $Y$ but smaller than $x_i$.}
 
 \medskip \noindent {\bf Semi-Join:}  Given two relations $R_1$ and $R_2$ with a common attribute
 $x$, the semijoin $R_1 \ltimes R_2$ returns all the tuples in $R_1$ whose value on $x$ matches that of at least one tuple in $R_2$. This can be reduced to a multi-search problem: For each $t \in R_1$, if its predecessor on the $x$ attribute in $R_2$ is the same as that of $t$ , then it is in the semijoin.

Note that we can remove all dangling tuples in an acyclic-join~\cite{yannakakis1981algorithms} by a constant number of semi-joins, so it can be done in $O(1)$ rounds with linear load.
 
\medskip\noindent{\bf Parallel-packing:} Given $\IN$ numbers $x_1, x_2, \cdots, x_{\IN}$ where $0 < x_i \le 1$ for $i =1,2, \cdots, {\IN}$, group them into $m$ sets $Y_1, Y_2, \cdots, Y_m$ such that $\sum_{i \in Y_j} x_i \le 1$ for all $j$, and $\sum_{i \in Y_j} x_i \ge {1 \over 2}$ for all but one $j$.  Initially, the $\IN$ numbers are distributed arbitrarily across all servers, and the algorithm should produce all pairs $(i,j)$ if $i \in Y_j$ when done. Note that $m \le 1+ 2\sum_i x_i$.

We are not aware of an explicit reference on this primitive, but it can be solved quite easily. Assume the input data is distributed across $p$ servers. We ask each server $i$ to first perform grouping on its local data.  It is obvious that the condition above can be satisfied.  The server then reports two numbers: $g_i$, the number of groups with sum between $1/2$ and $1$, and $h_i$, the sum of remaining group with sum smaller than $1/2$. Note that $g_i$ and $h_i$ can be $0$.  Next, we run the BSP algorithm for prefix-sums~\cite{goodrich11:_sortin} on the $g_i$'s.  After that, we can assign consecutive group id's to each of the $g_i$ groups on each server $i$.  For the remaining $p$ partial groups whose sums are $h_i$ with $0 < h_i < 1/2$, we recursively run the algorithm, using group id's starting from $\sum_i g_i +1$.  After the recursion returns, for each partial group $h_i$ that has been assigned to group $j$, we assign every element in $h_i$ to group $j$.  The problem size reduces by a factor of $\IN/p$ after each round, so the number of rounds is $O(\log_{\IN/p} \IN) = O(1)$.

 \medskip \noindent {\bf Server allocation~\cite{hu17:_output}:} 
 Assume each tuple has a subproblem id $j$, which identifies the subproblem it belongs to (the $j$'s do not have to be consecutive), and $p(j)$, which is the number of servers allocated to subproblem $j$. The goal is to attach to each tuple a range $[p_1(j), p_2(j)]$, such that the ranges of different subproblems are disjoint and $\max_j p_2(j)\le \sum_j p(j)$.  Thus, each tuple $t$ knows which servers have been allocated to the subproblem to which $t$ belongs.
 
 \medskip \noindent {\bf Computing the output size $\OUT$ of an acyclic join:}   This primitive is a special case of our join-aggregate algorithm, which will be described in Section~\ref{sec:join-aggregate}. 

 \input{hierarchical}
 \input{line3}
 \input{acyclic}
 \input{aggregate}
 \input{triangle}

 \bibliographystyle{abbrv}
 \bibliography{ms}

 \newpage
 \input{appendix}

 \end{document}

%% file: macros.tex
\usepackage[left=1in,top=1in,right=1in,bottom = 1in, nohead]{geometry}

\renewenvironment{thebibliography}[1]{%
	\begin{oldthebibliography}{#1}%
		\small
		\setlength{\parskip}{0.1ex}
		\setlength{\itemsep}{0.0ex}%
	}%
	{%
	\end{oldthebibliography}%
}
\usepackage{ifpdf}
\ifpdf                                                                  
\usepackage[breaklinks=true]{hyperref}
\hypersetup{
	colorlinks=false,               
	citebordercolor={0 1 0},        
	urlbordercolor={0 0 1},         
	linkbordercolor={1 0.2 0.2},    
	pdfborder={0 0 1.2},            
}
\else
\usepackage[hypertex, breaklinks=true]{hyperref}       
\hypersetup{
	colorlinks=true,               
	anchorcolor=yellow,
	linkcolor=green,
	filecolor=green,
	urlcolor=blue,
	citecolor=red,
}
\fi
\hypersetup{                                                            
	pdfpagemode=UseNone,               
	pdfstartview=Fit,
	pdfpagelayout=SinglePage,       
	bookmarks=true,
	bookmarksnumbered=true,
	bookmarksopenlevel=1,            
	bookmarkstype=toc,
	pdftex=true,
	unicode,
	plainpages=false,           
	hypertexnames=true,         
	pdfpagelabels=true,         
	hyperindex=true,
	naturalnames=false,                
}
\makeatletter
\def\@cite#1#2{[\textbf{#1\if@tempswa , #2\fi}]}
\def\@biblabel#1{[\textbf{#1}]}
\makeatother

\usepackage{mdframed}
\usepackage{graphicx}   
\usepackage{hyperref}   
\usepackage{amsmath,amssymb}    
\usepackage{thm-restate}
\usepackage{mdwlist}
\usepackage{mathtools}
\usepackage{xspace}
\usepackage{verbatim}   
\usepackage{xcolor}
\usepackage{makecell}
\usepackage[mathscr]{euscript}
\usepackage{tikz}
\usepackage{relsize}
\usepackage{booktabs}
\usepackage{tikz-qtree}
\usepackage{subcaption}
\usepackage{multirow}
\usetikzlibrary{arrows.meta}
\usetikzlibrary{fit,shapes,trees,shapes.geometric}
\usetikzlibrary{patterns,decorations.pathreplacing,calc}
\usetikzlibrary{matrix, positioning, arrows}
\usetikzlibrary{chains,shapes.multipart}
\usetikzlibrary{shapes,calc}
\usetikzlibrary{automata}
\usepackage[linesnumbered,algoruled, lined, noend]{algorithm2e}

\usepackage[normalem]{ulem}

\usepackage{amsthm}

\newcommand{\cut}[1]{}   
\makeatletter
\newcommand{\mytag}[2]{%
	\text{#1}%
	\@bsphack
	\protected@write\@auxout{}%
	{\string\newlabel{#2}{{#1}{\thepage}}}%
	\@esphack
}
\makeatother

\usepackage{aliascnt}

\newtheorem{corollary}{Corollary}

\newtheorem{lemma}{Lemma}

\newtheorem{theorem}{Theorem}

\providecommand{\E}[0]{\mathbf{E}}

%% file: hierarchical.tex
\section{r-Hierarchical Joins}
\label{sec:hierarchical}

Recall that in a hierarchical join, all attributes can be organized into a forest, such that $x$ is a descendant of $y$ if and only if $\E_x \subseteq \E_y$. Each $e\in \E$ corresponds to a node $x$ in the forest, such that $e$ contains precisely $x$ and all its ancestors.  A subclass of hierarchical joins are tall-flat joins.  For a tall-flat join, this attribute forest takes the form \revise{of} a special tree, which consists of a single ``stem'' plus a number of leaves at the bottom.  For example, $\revise{\Q_1}
	= R_1(x_1) \Join R_2(x_1, x_2) \Join R_3(x_1, x_2, x_3) \Join R_4(x_1, x_2, x_3, x_4) \Join R_5(x_1, x_2, x_3, x_5) \Join R_6(x_1, x_2,x_3,x_6)$ is a tall-flat join; $\revise{\Q_2}
	=R_1(x_1, x_2) \Join R_2(x_1, x_3, x_4) \Join R_3(x_1, x_3, x_5)$ is a hierarchical join (but not tall-flat).  Their attribute forests (actually, trees for these two cases) are shown in Figure~\ref{fig:hierarchical}.
\begin{figure}[h]
  \centering
  \includegraphics[scale = 1.1]{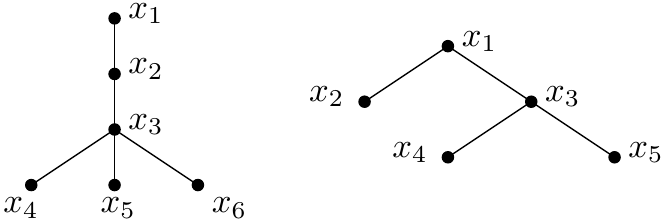}
  \caption{Examples of tall-flat and hierarchical join.}
  \label{fig:hierarchical}
\end{figure}
      
\revise{In this section, we study r-hierarchical joins.} A join is r-hierarchical if its reduced join is hierarchical.  For example, $\revise{\Q_2}
\Join R_4(x_3, x_5) \Join R_5(x_5)$ is an r-hierarchical join (but not hierarchical).  After an r-hierarchical join is reduced, its hyperedges must correspond to the leaves of the attribute forest.

\medskip
\subsection{BinHC algorithm revisited} 
We mentioned above that the HyperCube algorithm \cite{afrati11:_optim} is an instance-optimal algorithm for computing Cartesian products.  The BinHC algorithm \cite{beame14:_skew} is a generalization of the HyperCube algorithm to general joins.  For a join $\Q$, denote the residual query by removing attributes $\bm{x} \subseteq \V$ as $\Q_{\bm{x}}$. Let $\bm{u}$ be any fractional edge packing of $\Q_{\bm{x}}$ that saturates the attributes $\bm{x}$, i.e., $\sum_{e: x \in e}\bm{u}(e) \ge 1$ for every $x \in \bm{x}$, and $\sum_{e: x \in e}\bm{u}(e) \le 1$ for every $x \in  \V - \bm{x}$. Assuming knowing all degree information in advance, this algorithm computes $\Q$ on instance $\R$ in a single round with a load of $\widetilde{O}(\frac{\IN}{p} + L_{\textrm{BinHC}}(p,\R))$, where
\begin{equation*}
\label{eq:skew}
L_{\textrm{BinHC}}(p,\R) := \max_{\bm{x}, \bm{u}} \left(\frac{\sum_{\bm{a} \in \dom(\bm{x})} \prod_{e \in \E} |\sigma_{\bm{x} = \bm{a}} R(e)|^{\bm{u}(e)}}{p}\right)^{\frac{1}{\sum_{e \in \E} \bm{u}(e)}}
\end{equation*}
Here we define $0^0=0$. Note that for any $e\subseteq \bm{x}$, $|\sigma_{\bm{x} = \bm{a}} R(e)|$ is either $0$ or $1$, so we can just set $\bm{u}(e)=0$ for each such $e$ in the definition above.  


\begin{theorem}
	\label{the:tall-flat}
	On any tall-flat join and any instance $\R$, $L_{\textrm{BinHC}}(p,\R) = O\left(L_{\textrm{instance}}(p,\R)\right)$.
\end{theorem}

\begin{proof}
	Below, we write $L:=\Li(p, \R)$ to avoid notational clutter.  For an attribute set $\bm{x}$ and a fractional edge packing $\bm{u}$ of $\Q_{\bm{x}}$, define
	\[p(\bm{x}, \bm{u}) = \sum_{\bm{a} \in \dom(\bm{x})} \prod_{e \in \E} \left(\frac{|\sigma_{\bm{x} = \bm{a}} R(e)|}{L}\right)^{\bm{u}(e)}.\] To show $L_{\textrm{BinHC}}(p,\R) = O(L)$, it suffices to show that $p(\bm{x}, \bm{u}) = O(p)$ for all $\bm{x}$ and $\bm{u}$.
	
	Recall that in a tall-flat join, all attributes can be ordered as $x_1, x_2, \cdots, x_h, y_1, y_2, \cdots, y_l$ such that (1) $\E_{x_1} \supseteq \E_{x_2} \supseteq \cdots \supseteq \E_{x_h}$; (2) $\E_{x_h} \supseteq \E_{y_j}$ for $j =1, 2, \cdots, l$; (3) $|\E_{y_j}| = 1$ for $j =1, 2, \cdots, l$.
	Consider an attribute set $\bm{x} \subseteq \V$ under the following two cases.
	
	Case (1): $\{x_1, x_2, \cdots, x_h\} \subseteq \bm{x}$.  Consider any edge packing $\bm{u}$ of $\Q_{\bm{x}}$ that saturates $\bm{x}$ (in this case, we actually only need the fact $\bm{u}(e) \le 1$ for all $e$). As observed, we can eliminate any assignment $\bm{a} \in \dom(\bm{x})$ if there exists an edge $e \in \E$ such that $\sigma_{\bm{x} = \bm{a}} R(e) = \emptyset$, so it suffices to consider the remaining assignments $\bm{a^*} \in \dom(\bm{x})$ such that $\prod_{e \in \E} |\sigma_{\bm{x} = \bm{a^*}} R(e)| >0$. Then, we can bound $p(\bm{x}, \bm{u})$ as
	\begin{align*}
	p(\bm{x}, \bm{u})
	\le & \sum_{\bm{a^*}} \prod_{e \in \E} \left(\frac{|\sigma_{\bm{x} = \bm{a^*}} R(e)|}{L} + 1\right)^{\bm{u}(e)}  \\
	\le & \sum_{\bm{a^*}} \prod_{e \in \E} \left(\frac{|\sigma_{\bm{x} = \bm{a^*}} R(e)|}{L} + 1\right) \\
	= & \sum_{\bm{a^*}}\sum_{S \subseteq \E}{1\over L^{|S|}} \prod_{e \in S} |\sigma_{\bm{x} = \bm{a^*}} R(e)| \\
	= & \sum_{S \subseteq \E} \frac{1}{L^{|S|}} \sum_{\bm{a^*}}  \prod_{e \in S} |\sigma_{\bm{x} =  \bm{a^*}} R(e)| \\
	= & \sum_{S \subseteq \E} \frac{|\Q(\R,S)|}{L^{|S|}} = O(p).
	\end{align*}
	
	Case (2): There exists an $x_i \notin \bm{x}$. Let $i$ be the smallest such $i$.  Let $\bm{u}$ be any edge packing of $\Q_{\bm{x}}$.  In particular, we have $\sum_{e: x_i\in e} \bm{u}(e) \le 1$.  As observed earlier, we can set $\bm{u}(e)=0$ for any $e \subseteq \{x_1,\dots, x_{i-1}\}\subseteq \bm{x}$, so it suffices to consider the remaining edges.  Due to the tall-flat property, all these edges contain $x_i$.  Thus,
	\begin{align*}
	p(\bm{x}, \bm{u}) \le &\sum_{\bm{a}} \prod_{e: x_i \in e} \left(\frac{|\sigma_{\bm{x} = \bm{a}} R(e)|}{L}\right)^{\bm{u}(e)} \\
	\le & \prod_{e: x_i \in e} \left(\sum_{\bm{a}} \frac{|\sigma_{\bm{x} = \bm{a}} R(e)|}{L} \right)^{\bm{u}(e)} \quad \text{(H\"older's inequality)} \\
	\le& \prod_{e: x_i \in e} \left(\frac{|R(e)|}{L} \right)^{\bm{u}(e)} \\
	\le& \left(\frac{\IN}{L} \right)^{\sum_{e: x_i \in e} \bm{u}(e)} \le \frac{\IN}{L} \le p.
	\end{align*}
	Combining the two cases, the theorem is proved.
\end{proof}

\begin{theorem}
  \label{the:non-tall-flat}
  On any r-hierarchical join $\Q$ and instance $\R$ without dangling tuples, $L_{\textrm{BinHC}}(p,\R) = O\left(L_{\textrm{instance}}(p,\R)\right)$.
\end{theorem}

\begin{proof}
	Let $\T$ be the forest of attributes corresponding to $\Q$.  Consider an arbitrary attribute set $\bm{x}$.  We say that a root-to-leaf path in $\T$, which corresponds to some $e\in \E$, is {\em stuck} at the highest attribute on the path that is not included in $\bm{x}$.  In this way, all edges in $\Q$ can be divided into disjoint groups $\E_1, \E_2, \cdots \E_h$, such that edges in one group share the common stuck attribute.  Consider any fractional edge packing $\bm{u}$, we must have $\sum_{e \in \E_i} \bm{u}(e) \le 1$ for each $\E_i$ due to the packing constraint at the common stuck attribute of $\E_i$. Then, we can bound $p(\bm{x}, \bm{u})$ as
	\begin{align*}
	p(\bm{x}, \bm{u}) \le & \sum_{\bm{a}} \prod_{i \in \{1,2,\cdots,h\}}  \sum_{e \in \E_i} \left(\frac{|\sigma_{\bm{x} = \bm{a}} R(e)|}{L}\right)^{\sum_{e \in \E_i} \bm{u}(e)} \\
	\le & \sum_{\bm{a}} \prod_{i \in \{1,2,\cdots,h\}}  \sum_{e \in \E_i} \left(\frac{|\sigma_{\bm{x} = \bm{a}} R(e)|}{L} + 1\right) \\
	= & \sum_{\bm{a}} \prod_{i \in \{1,2,\cdots,h\}}  \left( \sum_{e \in \E_i} \frac{|\sigma_{\bm{x} = \bm{a}} R(e)|}{L} + |\E_i|\right) \\
	\le & \sum_{\bm{a}}  |\E|^{ |\E|} \sum_{S : |S \cap \E_i| \le 1, i=1,\dots,h}  \frac{\prod_{e \in S} |\sigma_{\bm{x} = \bm{a}} R(e)|}{L^{|S|}} \\
	= & |\E|^{ |\E|} \sum_{S : |S \cap \E_i| \le 1, i=1,\dots,h}  \frac{\sum_{\bm{a}}  \prod_{e \in S} |\sigma_{\bm{x} = \bm{a}} R(e)|}{L^{|S|}} \\
	\le & |\E|^{ |\E|} \sum_{S : |S \cap \E_i| \le 1, i=1,\dots,h}  \frac{|\Q(R,S)| }{L^{|S|}} = O(p).
	\end{align*}
	The last inequality needs some explanation: Any such $S$ includes at most one edge from each $\E_i$.  Thus, if two edges in $S$ share any common attribute, that attribute must be in $\bm{x}$ (otherwise they must belong to the same $\E_i$).  Thus, for any $\bm{a}$, all tuples in $\sigma_{\bm{x}=\bm{a}}R(e), e\in S$ join with each other, so we have \[\sum_{\bm{a}}\prod_{e \in S} |\sigma_{\bm{x} = \bm{a}} R(e)| = |\Join_{e \in S} R(e)|.\]  Furthermore, since there are no dangling tuples, every join result in $\Join_{e \in S} R(e)$ must be part of a full join result, so $|\Join_{e \in S} R(e)| \le |\Q(R,S)|$. 
\end{proof}

Note that since $\Li(p,\R)$ is a per-instance lower bound 
\revise{even} for multi-round algorithms, this means that the BinHC algorithm is instance-optimal even among all multi-round algorithms, up to polylogarithmic factors.  This result also incorporates the instance-optimality of the HyperCube algorithm on Cartesian products, which are special r-hierarchical joins without dangling tuples.

\paragraph{Remark}
Koutris and 
\revise{Suciu}~\cite{koutris11:_paral} show that non-tall-flat joins cannot be done with load $\tilde{O}({\IN \over p} + {\OUT \over p})$ by one-round algorithms. This does not contradict Theorem~\ref{the:non-tall-flat} since their lower bound construction uses dangling tuples.  Our result implies that the key barrier for one-round algorithms is actually the dangling tuples.  If they do not exist, one-round algorithms can go beyond tall-flat joins and solve r-hierarchical joins instance-optimally, up to polylog factors.  On the other hand, once $O(1)$ rounds are allowed, dangling tuples become irrelevant, since they can be removed with linear load and $O(1)$ rounds.

\subsection{An instance-optimal algorithm}
\label{sec:hierarchial-algorithm}

We have shown that the BinHC algorithm is an instance-optimal algorithm for r-hierarchical joins, but it has an instance-optimality ratio of $\log^{O(1)} p$, where the $O(1)$ exponent depends on the query size, and is at least $m$, the number of relations.  In this section, we improve the optimality ratio to $O(1)$, i.e., achieving a load of $O(\frac{\IN}{p} + \Li(p,\R))$.  Our algorithm uses $O(1)$ rounds, but note that BinHC also needs $O(1)$ rounds to remove the dangling tuples if they exist.  Furthermore, our algorithm is deterministic while BinHC is randomized.

As a preprocessing step, we remove all dangling tuples.  Then we reduce the join hypergraph, since if $e\subseteq e'$, $R(e)$ will not affect the final join results after dangling tuples are removed\footnote{Strictly speaking, this violates the tuple-based requirement that when emitting a join result, all the participating tuples must be present.  This can be easily fixed.  Before removing $R(e)$, we attach each tuple $t\in R(e)$ to all tuples in $R(e')$ that join with $t$.  This can be done by the multi-search primitive with linear load.}. Thus, we are left with a hierarchical join $\Q$ on an instance $\R$ with no dangling tuples.  

Let $\T$ be the attribute forest of $\Q$.  Recall that after the join is reduced, each relation corresponds to a leaf of $\T$, whose attributes are precisely the leaf's ancestors in $\T$.  Our algorithm is recursive.  We will show that the load of this algorithm is $O({\IN \over p} + \Li(p,\R))$ for any hierarchical join $\Q$ on any instance $\R$.  To simplify notation, we will not derive the exact constant in the big-Oh, which depends (exponentially) on the recursion depth.  Since the recursion depth is proportional to (actually, twice) the height of $\T$, which is a constant, this is not a concern.  Similarly, the number of servers employed by the algorithm will be $O(p)$, where the hidden constant may also depend on the recursion depth.

The base case is when $\Q$ has just one relation.  In this case the algorithm just emits all tuples in the relation, achieving the bound $O({\IN \over p} + \Li(p,\R))$ trivially.

For a general hierarchical join $\Q$ and an instance $\R$, we proceed as follows.  We first compute $L_{\textrm{instance}}(p,\R)$: We use $p$ servers to compute $|\Join_{e \in S} R(e)|$ for each $S \subseteq \E$ (recall that computing the output size of an acyclic join is an MPC primitive). This requires $O(p)$ servers with load $O({\IN \over p})$.  Note that $\Q(\R,S) = \Join_{e \in S} R(e)$ when there is no dangling tuples in $\R$, so we can compute $\Li(p,\R)$ as defined in (\ref{eq:instance}).  Setting $L = \frac{\IN}{p} + L_{\textrm{instance}}(p,\R)$, we will show below how to compute the join with $O(p)$ servers and load $O(L)$.

Let $k$ be the number of trees in $\T$.  We handle the following two cases using different recursive strategies:

\paragraph{Case (1): $k=1$} In this case, $\T$ is a tree.  Suppose the root attribute of $\T$ is $x$, which is included in all the relations.  Consider every $a \in \dom(x)$, and let $\R_a = \{\sigma_{x = a} R(e): e \in \E\}$.  It suffices to compute the residual query $\Q_x$ on each $\R_a$, but all the $\Q_x(\R_a)$'s have to be computed in parallel, using $O(p)$ servers in total.  Thus, the key is to allocate servers to these residual queries appropriately so as to ensure a uniform load of $O(L)$.  To do so, we first compute $\IN_a$, the input size of $\R_a$, for all $a\in \dom (x)$.  Since $\IN_a = \sum_{e \in \E} |\sigma_{x = a} R(e)|$, and each tuple belongs to exactly one $\R_a$, this is a sum-by-key problem, i.e., each tuple $t$ with $\pi_{x} t =a$ has key $a$ and weight $1$.  Note that $\IN = \sum_a \IN_a$.

An instance $\R_a$ is {\em heavy} if $\IN_a > L$ and {\em light} otherwise. We handle heavy and light instances in different ways.

\paragraph{Case (1.1): Light instances} We use the parallel-packing primitive to put the light instances into $O({\IN \over L}) = O(p)$ groups with each group having total input size $O(L)$.  Then we simply use one server to solve the instances in each group.  The load of each server is $O(L)$.  

\paragraph{Case (1.2): Heavy instances} By definition, there are at most ${\IN \over L} = O(p)$ heavy instances. For each heavy instance $\R_a$, we allocate $p_a = \lceil p \cdot \frac{\IN_a}{\IN}\rceil$ servers to compute in parallel the join size $|\Q_x(\R_a, S)|$ for all $a\in \dom(x)$ and all $S\subseteq \E$.  This uses $O(p)$ servers, and the load is $O(\max_a \frac{\IN_a}{p_a}) = O(\frac{\IN}{p})$.  Next, for each heavy instance $\R_a$, we allocate
\[p_a = \max_{S \subseteq \E} \frac{|\Q_x(\R_a, S)|}{L^{|S|}}\]
servers and compute $\Q_x(\R_a)$ recursively in parallel. The number of servers used is 
\[\sum_a p_a \le \sum_{a} \sum_{S \subseteq \E} \frac{|\Q_x(\R_a, S)|}{L^{|S|}} = \sum_{S \subseteq \E} \frac{|\Q(\R, S)|}{L^{|S|}}  = O(p).\]

By the induction hypothesis, computing $\Q_x(\R_a)$ with $p_a$ servers has a load of (the big-Oh of)
\begin{equation}
  \label{eq:2}
  \frac{\IN_a}{p_a} + \Li(p_a, \R_a) = {\IN_a \over p_a} + \max_{S \subseteq \E} \left(\frac{|\Q_x(\R_a, S)|}{p_a}\right)^{\frac{1}{|S|}}.
\end{equation}
We bound each term of (\ref{eq:2}): For a heavy instance $\R_a$, there must exist at least one $e \in \E$ such that $|\sigma_{x = a} R(e)| \ge \frac{1}{|\E|} \cdot \IN_a$. Furthermore, since there are no dangling tuples, every tuple in $\sigma_{x = a} R(e)$ must be part of a join result of $\Q_x(\R_a)$, so $|\sigma_{x = a} R(e)| = |\Q_x(\R_a, \{e\})|$.  Taking $S = \{e\}$, we have \[p_a \ge \frac{1}{L} \cdot |\Q_x(\R_a, \{e\})| = \frac{1}{L} \cdot |\sigma_{x = a} R(e)| \ge \frac{\IN_a}{|\E| \cdot L},\]
so ${\IN_a \over p_a} = O(L)$.  The second term of (\ref{eq:2}) is also bounded by $L$ simply by the definition of $p_a$. 

\paragraph{Case (2): $k>1$} In this case, the join becomes a Cartesian product $\Q_1(\R_1) \times \cdots \times \Q_k(\R_k)$, where each $\Q_i(\R_i)$ is a join under Case (1).  One would attempt to first compute each $\Q_i(\R_i)$ recursively, and then compute the Cartesian product, but this would not yield instance-optimality.  Just consider an instance with $|\Q_1(\R_1)|=1$ and $|\Q_2(\R_2)|=p\cdot \IN$, where $\Q_2(\R_2) = R_1(A, B) \Join R_2(B,C)$ with $\revise{|\dom(B)|=1}, |R_1| = \IN, |R_2| = p$.  On this instance, we have $\Li(p, \R)= \max({\IN \over p}, \sqrt{\IN})$, but if we took a two-step approach, merely storing the intermediate result $\Q_2(\R_2)$ would incur a load of $\Omega(\IN)$.  This means that we have to interleave the two steps so as to avoid storing the intermediate results $\Q_i(R_i)$ explicitly. 

We arrange servers into a $p_1 \times p_2 \times \cdots \times p_k$ hypercube, where the dimensions $p_1, p_2, \cdots, p_k$ will be determined later.  We identify each server with coordinates $(c_1, c_2, \cdots, c_k)$, where $c_i \in [p_i]$.  For every combination $c_1, \dots, c_{i-1}, c_{i+1}, \dots, c_k$, the $p_i$ servers with coordinates $(c_1, \cdots, c_{i-1}, *, c_{i+1}, \cdots, c_k)$ form a group to compute $\Q_i(\R_i)$ (using the algorithm under Case (1)).  Yes, each $\Q_i(\R_i)$ is computed $p_1\cdots p_{i-1}p_{i+1}\cdots p_k$ times, which seems to be a lot of redundancy.  However, as we shall see, there will be no redundancy in terms of the final join results, and it is exactly due to this redundancy that we avoid the shuffling of the intermediate result and achieve an optimal load.  Consider a particular server $(c_1,\dots, c_k)$.  It participates in $k$ groups, one for each $\Q_i(\R_i), i=1,\dots,k$.  For each $\Q_i(\R_i)$, it emits a subset of its join results, denoted $\Q_i(\R_i, c_1\dots, c_k)$.  Then the server emits the Cartesian product $\Q_1(\R_1, c_1\dots, c_k)\times \cdots \times \Q_k(\R_k, c_1\dots, c_k)$.  Note that for each group of servers computing $\Q_i(\R_i)$, the $p_i$ servers in the group emit $\Q_i(R_i)$ with no redundancy, so there is no redundancy in emitting the Cartesian product.

It remains to show how to set $p_1,\dots, p_k$ so that $p_1\cdots p_k = O(p)$ and each server has a load of $O(L)$.  To do so, we first compute $\IN_i$, the input size of $\R_i$, in the same way as in Case (1).  An instance $\R_i$ is {\em heavy} if $\IN_i > L$ and {\em light} otherwise. For each heavy instance $\R_i$, we use $p$ servers to compute $|\Join_{e \in S} \R_i(e)| = |\Q_i(R_i, S)|$ for all $S \subseteq \E_i$, where $\E_i$ is the set of edges in $\Q_i$. This requires $O(p)$ servers with load $O(\frac{\IN}{p})$.  Then if $\R_i$ is light, we set $p_i = 1$; otherwise set \[p_i = \max_{S \subseteq \E_i} \left\lceil \frac{|\Q_i(\R_i, S)|}{L^{|S|}}\right \rceil.\] Let $I=\{ i \mid \R_i \text{ is heavy}\}$. The number of servers used is
\begin{align*}
\prod_{i \in I} p_i \le & \prod_{i \in I} \sum_{S \subseteq \E_i} ( \frac{|\Q_i(\R_i, S_i)|}{L^{|S|}} + 1 ) 
\le \sum_{S \subseteq \bigcup_{i \in I} \E_i} \frac{|\Q(\R, S)|}{L^{|S|}} = O(p).
\end{align*}

Finally, consider the load of each server, which serves to compute each $\Q_i(\R_i)$ with a group of $p_i$ servers.  For a light $\R_i$, $p_i=1$ and it imposes a load of $O(L)$.  For a heavy $\R_i$, by the induction hypothesis, the load is (the big-Oh of)
\begin{equation*}
{\IN_i \over p_i} + \Li(p_i, \R_i) = {\IN_i \over p_i} + \max_{S \subseteq \E_i}\ \left(\frac{|\Q_i(\R_i, S)|}{p_i}\right)^{\frac{1}{|S|}}.
\end{equation*}
This can be bounded by $O(L)$ using the same argument as Case (1.2).  Summing over all $i=1,\dots, k$ increases the load by just a $k=O(1)$ factor.  

The induction proof thus completes and we obtain the following result.
\medskip
\begin{theorem}
	\label{the:hierarchical}
        On any r-hierarchical join query $\Q$ and any instance $\R$, there is an algorithm computing $\Q(\R)$ in $O(1)$ rounds with load $O(\frac{\IN}{p} + L_{\textrm{instance}}(p,\R))$.
\end{theorem}

Since an instance-optimal algorithm is also output-optimal, we also obtain an output-optimal algorithm for r-hierarchical joins.  In fact, we can derive a closed-form formula of the output-optimal bound, i.e., we bound 
\begin{align*}
\max_{\R \in \mathfrak{R}(\IN, \OUT)} \Li(p, \R) 
\end{align*} as a function of $\IN$ and $\OUT$.  First, observe that $\Li(p, \R)$ only depends to the reduced instance of $\R$, so we can assume that $\R$ contains no dangling tuples. Then, we can rewrite $\max_{\R \in \mathfrak{R}(\IN, \OUT)} \Li(p, \R)$ as 
\begin{align*}
\max_{\R \in \mathfrak{R}(\IN, \OUT)} \Li(p, \R) = & \max_{\R \in \mathfrak{R}(\IN, \OUT)} \max_{S \subseteq \E} \ (\frac{|\Q(\R, S)|}{p} )^{\frac{1}{|S|}} \\
=& \max_{\R \in \mathfrak{R}(\IN, \OUT)} \ \max_{S \subseteq \E} \  (\frac{|\Join_{e \in S} R(e)|}{p} )^{\frac{1}{|S|}}\\
=& \max_{S \subseteq \E} \max_{\R \in \mathfrak{R}(\IN, \OUT)}  (\frac{|\Join_{e \in S} R(e)|}{p} )^{\frac{1}{|S|}} 
\end{align*}
Consider a specific subset $S \subseteq \E$ and an arbitrary instance $\R \in \mathfrak{R}(\IN, \OUT)$. One trivial upper bound for $|\Join_{e \in S} R(e)|$ is $\OUT$. The other bound is $\IN^{|S|}$ when the join degenerates to a Cartesian product.  With these observations, we can bound the quantity above as:
\begin{align*}
\max_{S \subseteq \E} \max_{\R \in \mathfrak{R}(\IN, \OUT)}  (\frac{|\Join_{e \in S} R(e)|}{p} )^{\frac{1}{|S|}}  = & \max_{k \in [n]} \max_{S \subseteq \E: |S|=k} \max_{\R \in \mathfrak{R}(\IN, \OUT)}  (\frac{|\Join_{e \in S} R(e)|}{p} )^{\frac{1}{|S|}} \\
\le & \max_{k \in [n]} \min \{\frac{\IN}{p^{\frac{1}{k}}},  (\frac{\OUT}{p} )^{\frac{1}{k}} \} 
= \max \{\frac{\IN}{p^{\frac{1}{k^*-1}}},  (\frac{\OUT}{p} )^{\frac{1}{k^*}} \}
\end{align*}
where $k^*$ denotes the integer $\lceil \log_{\IN} \OUT \rceil$.

Next, we show that this is tight, i.e., there exists an instance $\R \in \mathfrak{R}(\IN, \OUT)$ such that for one subset $S_1 \subseteq \E$ involving $k^*-1$ relations, there is $|\Join_{e \in S_1} R(e)| = \IN^{k^*-1}$ and for another subset $S_2 \subseteq \E$ involving $k^*$ relations, there is $|\Join_{e \in S_2} R(e)| = \OUT$. Our hard instance construction is based on the following property of acyclic joins (this is probably known, but we cannot find an explicit reference):

\begin{lemma}
	\label{lem:acyclic-integral}
	An acyclic join has integral edge cover number.
\end{lemma}

\begin{proof}
	For an acyclic hypergraph $\Q= (\V, \E)$, denote the optimal edge covering of $\Q$ as $\mathcal{C}$. If there exist $e, e' \in \E$ such that $e \subseteq e'$, then $\mathcal{C}(e) = 0$; otherwise we can just shift the weight from $e$ to $e'$ and obtain a better (at least not worse) edge covering. So the optimal edge cover of $\Q$ is equivalent to that of the residual query by removing $e$. If there exists an attribute that appears only in edge $e$, then $\mathcal{C}(e) = 1$. So the optimal edge cover of $\Q$ is equivalent to the edge $e$ and the optimal edge cover of the residual query by removing all attributes in $e$. After recursively apply these two procedures, the query will become an empty set implied by the GYO reduction~\cite{abiteboul1995foundations}. In this process, every edge chosen by $\mathcal{C}$ has weight $1$. 
\end{proof}

Let $\mathcal{C}$ be the optimal edge covering of $\Q$. We identify two subsets of $\mathcal{C}$ with $k^*-1, k^*$ edges respectively, denoted as $\mathcal{C}_{k^*-1}, \mathcal{C}_{k^*}$, such that $C_{k^*-1} \subseteq C_{k^*}$. Such two subsets can always be found since $|C| \ge \lceil \log_\IN \OUT \rceil$ by the AGM bound~\cite{atserias2008size}. We consider a hard instance $\R$ constructed as below. Each edge $e \in \mathcal{C}$ is associated with at least one unique attribute denoted as $e(u)$. One of the unique attributes in $e$ for $e \in \mathcal{C}_{k^*}$ has $\IN$ distinct values in its domain while one of the unique attributes in $e$ for $e \in \mathcal{C}_{k^*} - \mathcal{C}_{k^*-1}$ has $\frac{\OUT}{\IN^{k^*-1}}$ distinct values in its domain. Remaining attributes have only one value in their domains. On this instance, there is $|\Join_{e \in \mathcal{C}_{k^*-1}} R(e)| = \IN^{k^*-1}$ and $|\Join_{e \in \mathcal{C}_{k^*}} R(e)| = \OUT$.

\begin{theorem}
	\label{lem:instace-output}
	There is an algorithm that computes any r-hierarchical join in $O(1)$ rounds with load $O\left(\frac{\IN}{p^{1/{\max\{1, k^*-1\}}}} + (\frac{\OUT}{p})^{\frac{1}{k^*}}\right)$, where $k^* = \lceil \log_{\IN} \OUT \rceil$.  This bound is output-optimal.
\end{theorem}

Below we give a cleaner output-sensitive bound.  This is not tight for $\OUT > \IN^2$, but easier to use.  In particular, this result will be used in the analysis of the output-sensitive algorithm for arbitrary acyclic joins in Section~\ref{sec:acyclic-algorithm}.

\begin{corollary}
	\label{cor:instace-output}
	There is an algorithm that computes any r-hierarchical join in $O(1)$ rounds with load $O(\frac{\IN}{p} + \sqrt{\frac{\OUT}{p}})$.
\end{corollary}

\begin{proof}
	When $\OUT \le \IN$, we have $k^* =1$ and the load complexity is $O(\frac{\IN}{p})$ trivially. For $k^* \ge 2$, the term $(\frac{\OUT}{p})^{{1/k^*}}$ is always no larger than $\sqrt{\frac{\OUT}{p}}$. The term ${\IN/p^{1/{\max\{1, k^*-1\}}}}$ is also no larger than $\sqrt{\frac{\OUT}{p}}$ as long as $\IN^2 \cdot p^{1-{2/\max\{1, k^*-1\}}} \le \OUT$, which always holds when $k^* \ge 2$.
\end{proof}


%% file: line3.tex
\section{Line-3 Join}
\label{sec:line3}

The simplest acyclic but not r-hierarchical join is the line-3 join $R_1(A, B) \Join R_2(B, C) \Join R_3(C, D)$.  In this section, we give an output-optimal MPC algorithm with load $O(\frac{\IN}{p} + \frac{\sqrt{\IN \cdot \OUT}}{p})$, together with a matching lower bound.  In particular, the lower bound implies that instance-optimal algorithms are not possible for the line-3 join.  In Section~\ref{sec:full-acyclic}, we extend these results to arbitrary acyclic joins.

\subsection{The Yannakakis algorithm revisited}

The Yannakakis algorithm first removes all the dangling tuples, which is just a series of semi-joins and can be done with load $O({\IN \over p})$.  Then the algorithm performs pairwise joins in some arbitrary order.  In the RAM model, the join order does not affect the asymptotic running time: After dangling tuples have been removed, any intermediate join result is part of a full join result, so the running time of the last join, which is $\Theta(\OUT)$, dominates that of any intermediate join.  In fact, this argument applies on a per-instance basis, and the Yannakakis algorithm is instance-optimal on any instance with any join order.

Interestingly, the join order does matter in the MPC model.  Consider the following instance of the line-3 join (see the top half of Figure~\ref{fig:l3-instance}).  Attributes $A,B,C,D$ have domain sizes $\frac{\OUT}{N}, \frac{N^2}{\OUT}, N, 1$, respectively. Set $R_1(A,B) = \dom(A) \times \dom(B)$, $R_2(B,C)$ is a one-to-many relation from $\dom(B)$ to $\dom(C)$, and $R_3(C,D) = \dom(C) \times \dom(D)$.  Note that this instance has $\IN = \Theta(N)$ and the output size is exactly $\OUT$.  Consider first the join plan $(R_1 \Join R_2) \Join R_3$, and note that $|R_1\Join R_2| = |R_1 \Join R_2 \Join R_3| = \OUT$. Using the $O({\IN \over p} + \sqrt{\OUT \over p})$-load algorithm \cite{beame14:_skew,hu17:_output} for binary joins, the load of computing $R_1 \Join R_2$ is $O({\IN \over p} + \sqrt{{\OUT \over p}})$.  However, since the output of the first join is the input of the second join, the input size for the second join is $\OUT$, so the load of the second join is $O({\OUT \over p} + \sqrt{{\OUT \over p}}) = O({\OUT \over p})$.  In general, the intermediate join result can be as large as $O(\OUT)$, which is why the Yannakakis algorithm incurs a load of $O({\OUT \over p})$ (after dangling tuples are removed) on an acyclic join, as observed in \cite{afrati2014gym,koutris18:_algor}.

Now consider the alternative plan $R_1 \Join (R_2 \Join R_3)$. Note that $|R_2 \Join R_3| = O(\IN)$, so the load of computing $R_2\Join R_3$ is $O({\IN \over p})$, while the load of computing the second join is $O({\IN \over p} + \sqrt{{\OUT \over p}})$.  Crucially, the reason why the second plan is better is that it has a smaller intermediate join size.  Note that a smaller intermediate join size does not matter in the RAM model, where the total cost is always dominated by the last join. But it does matter in the MPC model, because of the $O({\IN \over p} + \sqrt{{\OUT \over p}})$ load complexity of a binary join, which has a linear dependency on the input size but sublinear in the output size.  Fundamentally, this is because the MPC model is all about {\em locality}: algorithms strive to send all ``related'' tuples to the same server so as to maximize the number of join results that can be found by the server locally.

\begin{figure}[h]
	\centering
	\includegraphics[scale=1.0]{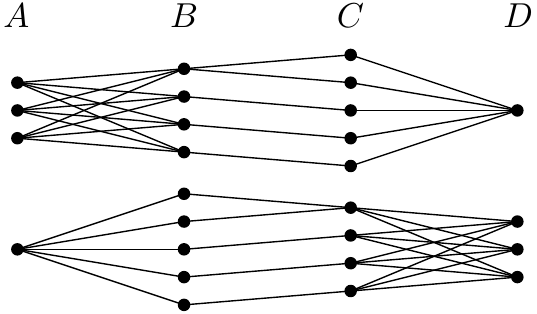}
	\caption{A hard instance for the Yannakakis algorithm.}
	\label{fig:l3-instance}
\end{figure}

Now, the key question is if there is always a join plan with an intermediate join size asymptotically smaller than $O(\OUT)$.  Unfortunately, the answer is no.  A bad example can be easily constructed, by just putting two of the above instances together, but in opposite directions (see Figure~\ref{fig:l3-instance}).  Nevertheless, this bad example precisely points us to the right direction: Although a global best join order may not exist, but if we decompose the join into multiple pieces, it is possible to find a provably good join order for each.  This is exactly the basic idea of our algorithm, presented next.

\subsection{A new algorithm for the line-3 join}

We first compute $\OUT$ (an MPC primitive).  Then we proceed in two steps:

\paragraph{Step (1): Computing degrees} 
For a value in attribute $B$, it is {\em heavy} if its {\em degree} in relation $R_1$, i.e., $|\sigma_{B=b} R_1|$, is greater than $\tau$ (value to be determined later), otherwise {\em light}.  We first use the sum-by-key primitive to compute the degrees of all $b$'s for $b \in \dom(B)$.  After classifying the values in $\dom(B)$ as heavy and light, we divide tuples in $R_1$ and $R_2$ also into heavy tuples and light tuples, depending on their $B$ value.  More precisely, a tuple in $R_1$ or $R_2$ is heavy if its $B$ value is heavy, and light otherwise.  This can be done by the multi-search primitive.  We denote the heavy (resp.\ light) tuples in $R_i$ as $R^H_i$ (resp. $R^L_i$), for $i=1,2$. 

\paragraph{Step (2): Decomposing the join}
We decompose the join into the following two parts, and compute them using different join orders:
\begin{align*}
  \Q_1 & = R_1^H \Join (R_2^H \Join R_3),\\
  \Q_2 & = (R_1^L \Join R_2^L) \Join R_3.
\end{align*}
Note that since $R_1$ and $R_2$ are both divided according to the $B$ attribute, $R_1^H$ do not join with $R_2^L$, $R_1^L$ do not join with $R_2^H$.

\medskip Now we analyze the load. For $\Q_1$, the intermediate join $R_{23} = R_2^H \Join R_3$ has size bounded by $\frac{\OUT}{\tau}$, since each intermediate join result from $R_{23}$ has a heavy $B$ value, so it joins with at least $\tau$ tuples in $R_1$.  Thus, the load of computing $\Q_1$ is (big-Oh of)
\begin{equation}
  \label{eq:3}
  {\IN \over p} + {\OUT \over p\tau} + \sqrt{\OUT \over p}.
\end{equation}

For $\Q_2$, the intermediate join $R_{12} = R_1^L \Join R_2^L$ has size bounded by $\IN \cdot \tau$, since each light tuple from $R_2$ can join with at most $\tau$ tuples from $R_1$. Thus, the load of computing $\Q_2$ is (big-Oh of) \revise{$O({\IN \over p} + {\IN \cdot \tau \over p} + \sqrt{\OUT \over p})$.}
\begin{equation}
  \label{eq:4}
  {\IN \over p} + {\IN \cdot \tau \over p} + \sqrt{\OUT \over p}.
\end{equation}

Setting $\tau = \sqrt{{\OUT \over \IN}}$ balances the second term in (\ref{eq:3}) and in (\ref{eq:4}), and we obtain the claimed result (note that $\sqrt{{\OUT \over p}} \le {\sqrt{\IN \cdot \OUT} \over p}$ for $\IN \ge p$):
\begin{theorem}
	\label{the:hybrid-yannakakis}
There is an algorithm computing the line-3 join with load $O\left(\frac{\IN}{p} + \frac{\sqrt{\IN \cdot \OUT}}{p}\right)$ in $O(1)$ rounds.
\end{theorem}

\subsection{Lower bound}

We prove the following lower bound on any tuple-based algorithm for computing the line-3 join. 

\begin{theorem}
  \label{lem:3-line-lower}
  For any $\OUT \ge \IN$, there exists an instance $\R$ for the line-3 join with input size $\Theta(\IN)$ and output size $\Theta(\OUT)$, such that any tuple-based algorithm computing the join in $O(1)$ rounds must have a load of $\Omega\left(\min\left\{\frac{\sqrt{\IN \cdot \OUT}}{p \cdot \log \IN}, \frac{\IN}{\sqrt{p}}\right\}\right)$.
\end{theorem}

\begin{proof}
	Our lower bound argument is combinatorial in nature.  We will construct a hard instance $\R$, such that a server can produces at most $J(L)$ join results in a round, no matter which $L$ tuples from $\R$ are loaded to the server.  Then $p$ servers can product at most $O(p\cdot J(L))$ results over $O(1)$ rounds. Setting $p\cdot J(L) = \Omega(\OUT)$ will yield a lower bound on $L$.  Thus, any upper bound on $J(L)$ will yield a lower bound on $L$, and we will only focus on upper bounding $J(L)$.
	
	We construct $\R$ using the probabilistic method, i.e., we randomly generate an instance, and show that with positive probability (actually, with high probability), such a randomly generated instance satisfies our needs.  The construction is similar to the one used in \cite{hu17:_output}, but the parameters and arguments are different.
	
	A randomly constructed instance is shown in Figure~\ref{fig:l3-instance-lower}.  In fact, only $R_2$ is random, while $R_1$ and $R_3$ are deterministic. Let $N = \frac{\IN}{3}, \tau = \sqrt{\frac{\OUT}{N}}$, and set $\dom(B) = \dom(C) = {N \over \tau}$.  Each distinct value of $B$ appears in $\tau$ tuples in $R_1(A,B)$, and each distinct value in $C$ appears in $\tau$ tuples in $R_3(C,D)$.  The $\tau$ tuples in $R_1$ (resp. $R_3$) that share the same $B$ (resp. $C$) value are called a {\em group}.  For each pair of values $(b,c), b\in \dom(B), c\in\dom(C)$, the tuple $(b,c)$ is included in $R_2(B,C)$ with probability $\frac{\tau^2}{N}$ independently.  Note that $|R_1| = |R_3| = N$, and $E[|R_2|] = N$, so the input size is expected to be $\IN$. The output size is expected to be $\tau^2 \cdot ({N \over \tau})^2 \cdot {\tau^2 \over N} = \OUT$.  By the Chernoff inequality, the probability that the input size or output size deviates from their expectations by more than a constant fraction is $\exp(-\Omega(N))$.
	\begin{figure}[h]
		\centering
		\includegraphics[scale=1.0]{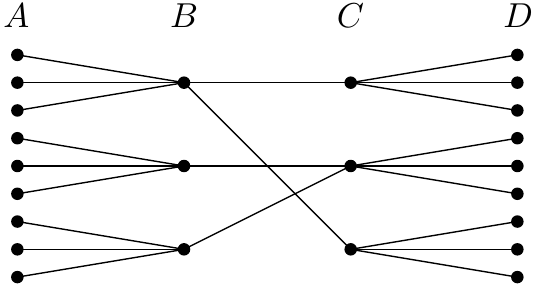}
		\caption{A randomly constructed instance on $L_3$.}
		\label{fig:l3-instance-lower}
	\end{figure}
	
	To give an upper bound on $J(L)$, we only restrict the server to load at most $L$ tuples from $R_1$ and $R_3$, while tuples in $R_2$ can be accessed for free.  Furthermore, we argue below that we only need to consider the situation where the server loads $R_1$ and $R_3$ in whole groups.  Suppose two groups in $R_1$, say, $g_1$ and $g_2$, are not loaded in full (we may assume w.l.o.g. that $L$ is a multiple of $\tau$, so there cannot be exactly one non-full group): $x_1 < \tau$ tuples of $g_1$ and $x_2 < \tau$ tuples of $g_2$ have been loaded.  Suppose they respectively join with $y_1$ and $y_2$ tuples in $R_3$ that are loaded by the server.  Note that they will produce $x_1 y_1 + x_2 y_2$ join results.  Without loss of generality, assume $y_1 \ge y_2$.  Now consider the alternative where the server loads $x_1 + 1$ tuples of $g_1$ and $x_2-1$ tuples of $g_2$.  Then this would produce $(x_1 +1)y_1 + (x_2-1)y_2 = x_1 y_1 +x_2y_2 +y_1 -y_2 \ge x_1 y_1+ x_2y_2$ tuples.  This means that by moving one tuple from $g_2$ to $g_1$, the server can only get more join results (at least not less).  We can move tuples from one group to another as long as there are two non-full groups.  Eventually we arrive at a situation where all groups of $R_1$ are loaded by the server in full, without decreasing the reported join size.  Next, we apply the same transformation to the groups of $R_3$ to make all its groups full as well.  Therefore, to maximize $J(L)$, the server should only load $R_1$ and $R_3$ in full groups.
	
	Thus, the server loads $\frac{L}{\tau}$ groups from $R_1$ and $\frac{L}{\tau}$ groups from $R_3$.  Below we show that a random instance constructed as above has the following property with high probability: On every possible choice of the ${L \over \tau}$ groups of $R_1$ and ${L \over \tau}$ groups of $R_3$ to be loaded, $J(L)$ is always bounded.
	
	Consider a particular choice of the $\frac{L}{\tau}$ groups from $R_1$ and $\frac{L}{\tau}$ groups from $R_3$ to be loaded. There are $\left(\frac{L}{\tau}\right)^2$ pairs of groups, and each pair has probability $\frac{\tau^2}{N}$ to join, so we expect to see $\frac{L^2}{N}$ pairs to join. Because the pairs join independently, by the Chernoff bound, the probability that more than $\delta \cdot \frac{L^2}{N}$ pairs join is at most $\exp\left(-\Omega(\delta \cdot \frac{L^2}{N})\right)$, for some parameter $\delta \ge 2$ to be determined later. There are $O\left((\frac{N}{\tau})^{\frac{2L}{\tau}}\right)$ different choices of $\frac{L}{\tau}$ groups from $R_1$ and ${L \over \tau}$ groups from $R_3$. So, by the union bound, the probability that one of them yields more than $\delta \cdot \frac{L^2}{N}$ joining groups is at most
	\begin{align*}
	&  O\left(({N \over \tau})^{\frac{2L}{\tau}}\right) \cdot \exp\left(-\Omega(\frac{\delta L^2}{N})\right) = \exp\left(-\Omega(\frac{\delta L^2}{N}) + O(\frac{L}{\tau} \log N)\right).
	\end{align*}
	This probability is exponentially small if $\delta \cdot \frac{L^2}{N} > c_1 \cdot \frac{L}{\tau}  \log N$ for some sufficiently large constant $c_1$, so we set
	\begin{equation}
	\label{eq:delta}
	\delta = \max\left\{\frac{c_1 \cdot N \log N}{\tau L}, 2\right\}.
	\end{equation}
	
	Since each joining group produces $\tau^2$ join results, we have shown that with high probability, a random instance has the property that no matter which $L$ tuples are loaded, we always have   $J(L) \le \delta \cdot \frac{\tau^2L^2}{N}$.  Putting this into $p\cdot J(L) = \Omega(\OUT)$, we obtain
	\begin{equation}
	\label{eq:counting}
	\delta \cdot \frac{\tau^2 L^2}{N} \cdot p = \Omega(\OUT).
	\end{equation}	
	Plugging (\ref{eq:delta}) into (\ref{eq:counting}), we have
	\[ \max\left\{\frac{N \log N}{\tau L} \cdot \frac{\tau^2 L^2}{N} \cdot p,  \frac{\tau^2 L^2}{N} \cdot p \right\} = \Omega(\OUT),\]
	or
	\begin{equation}
	\label{eq:5}
	\max\left\{ \tau p L \log N, {\tau^2 p L^2 \over N}  \right\} = \Omega(\OUT).
	\end{equation}
	Plugging in $\tau = \sqrt{\frac{\OUT}{N}}$, $N = {\IN \over 3}$,
	\[\max\left\{\log \IN \cdot p L \sqrt{{\OUT \over \IN}},  {\OUT \cdot pL^2\over \IN^2} \right\} = \Omega(\OUT).\]
	The theorem is then proved after rearranging the terms.
\end{proof}

Ignoring logarithmic factors, this lower bound completes our understanding of the line-3 join in \revise{terms} of output-optimality: (1) When $\OUT \le \IN$, the Yannakakis algorithm has linear load $O\left(\frac{\IN}{p}\right)$.  (2) When $\IN < \OUT \le p \cdot \IN$, the lower bound becomes $\tilde{\Omega}\left(\frac{\sqrt{\IN \cdot \OUT}}{p}\right)$, which is matched by our new algorithm.  (3) When $\OUT \ge p\cdot \IN$, the lower bound is $\Omega\left({\IN \over \sqrt{p}}\right)$, which is matched by the worst-case optimal algorithm in~\cite{hu2016towards,koutris16:_worst}.  In particular, this means that when $\OUT$ is large enough, the load complexity of the join is no longer output-sensitive. This also stands in contrast with the RAM model, where the complexity of any acyclic join always grows linearly with $\OUT$.

An easy corollary 
is the following result, which shows that instance-optimality is not achievable for the line-3 join.

\begin{corollary}
  \label{cor:instance-impossible}
  For any $\IN \ge p^{3/2}$, there is an instance $\R$ with input size $\Theta(\IN)$ for the line-3 join, such that any tuple-based algorithm computing the join in $O(1)$ rounds must have a load of $\Omega({\IN \over p^{1/2} \log\IN})$, while $\Li(p, \R)=O({\IN \over p})$.
\end{corollary}

\begin{proof}
	We use $\OUT = p\cdot \IN, \tau = \sqrt{p}$ in the lower bound construction above.  Plugging these values into (\ref{eq:5}), we obtain the claimed lower bound. On the other hand, we have $ \Li(p,\R) $ as large as
	\begin{align*}
	\max\{ {\IN \over p}, ({\tau \IN \over p})^{1/2}, ({\OUT \over p})^{1/3} \}  = \max \{ {\IN \over p}, {\IN^{1/2} \over p^{1/4}}, \IN^{1/3}\}.
	\end{align*}
	As long as $\IN \ge p^{3/2}$, the first term dominates.  
\end{proof}



%% file: acyclic.tex
\section{Acyclic Joins}
\label{sec:full-acyclic}

In this section, we first extend the results from the previous section to arbitrary acyclic joins.  Specifically, the algorithm is a (nontrivial) generalization of the line-3 algorithm, but it is self-contained; the lower bound builds on top of the hard instance of the line-3 join.  

\subsection{Algorithm}
\label{sec:acyclic-algorithm}

As a preprocessing step, we remove all dangling tuples.  We also assume that the output size $\OUT$ has been computed \revise{(an MPC primitive)}.

Recall that in an acyclic join $\Q = (\V, \E)$, the hyperedges $\E$ can be organized into a join tree $\T$, such that for each attribute $x\in \V$, the nodes corresponding to $\E_x$ are connected in $\T$.  Given such a join tree $\T$, our algorithm recursively decomposes the join into multiple pieces, and apply a different join strategy for each.

We start from an internal node of $\T$ whose children are all leaves.  Let this node be $e_0$, which has $k$ leaf children $e_1, \cdots, e_k$ (see Figure~\ref{fig:alpha} for an example). Let $s_i = e_0 \cap e_i$ be the set of join attributes between $e_0$ and $e_i$. We will assume $s_i\ne \emptyset$; otherwise we can add a dummy attribute to both $e_0$ and $e_i$ and all tuples in $R(e_0)$ and $R(e_i)$ share the same value on this dummy attribute (e.g., we add a dummy attribute $H'$ to both $e_0$ and $e_6$ in Figure~\ref{fig:alpha}).  Note that the join tree ensures the property that if $x\in e_i \cap e_j$ for $i \neq j$, then $x \in e_0$.

\begin{figure}[h]
	\centering	
	\includegraphics[scale = 0.75]{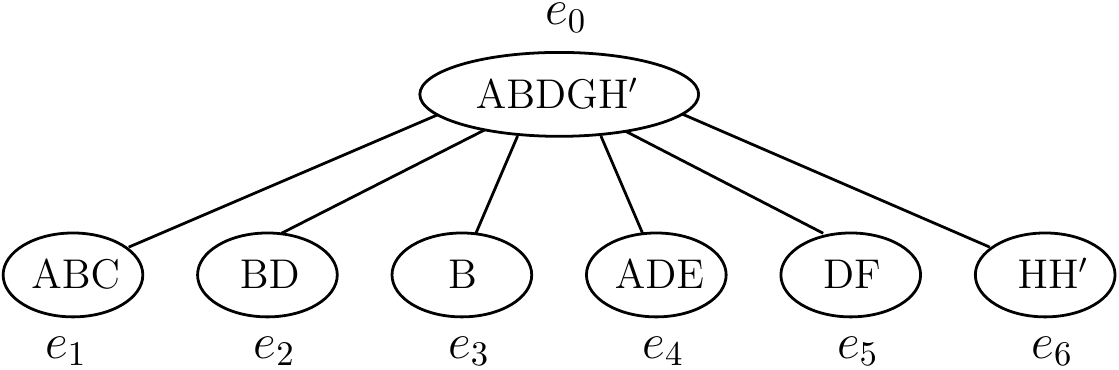}
	\caption{A node $e_0$ in the join tree $\T$ and its leaf children $e_1, e_2, e_3, e_4, e_5, e_6$.}
	\label{fig:alpha}
\end{figure}

Let $N_{\alpha} = \sum_{i=1}^k |R(e_i)|$ and $N_{\beta} = \IN - N_\alpha$.  We will actually prove a slightly tighter bound, that the load of our algorithm is bounded by $O({\IN \over p} + {\sqrt{N_\beta \cdot \OUT} \over p} + \sqrt{\OUT \over p})$. 

Set $\tau = \sqrt{\frac{\OUT}{N_{\beta}}}$. Our algorithm proceeds in three steps. 

\paragraph{Step (1): Computing data statistics} 
In each relation $R(e_i)$, $i=1,\dots,k$, let $v$ be an assignment of values for attributes $s_i$. The set of {\em heavy} assignments in $R(e_i)$ is 
\[H(s_i, e_i) = \{v \in \pi_{s_i} R(e_i): |\sigma_{s_i = v} R(e_i)| \ge \tau\}.\]
Tuples in $R(e_i)$ can also be identified as {\em heavy} or {\em light}, depending on their projection on attributes $s_i$. More precisely, a tuple $t \in R(e_i)$ is heavy if $\pi_{s_i} t \in H(s_i, e_i)$. The set of heavy tuples and light tuples in $R(e_i)$ are denoted as $R_H(e_i)$ and $R_L(e_i)$, respectively.  All the statistics can be computed in by the sum-by-key and multi-search primitives with linear load.

Let $\bar{\E} = \E - \{e_0, e_1, \cdots, e_k\}$.  We decompose the join into the following sub-joins:
\[R(e_0) \Join R_{?}(e_1) \Join \cdots \Join R_{?}(e_k) \Join \left(\Join_{e \in \bar{\E}} R(e)\right),\]
where each $?$ can be either $H$ or $L$. Note that there are $2^k$, which is a constant, sub-joins, so we can afford to use $p$ servers for each sub-join.  If a sub-join involves at least one $R_H(e_i)$, we apply the procedure in step (2) to it. In step (3), we handle the case where all $?$ are $L$.   

\paragraph{Step (2): Sub-joins with at least one $R_H(e_i)$} 
Without loss of generality, suppose $R_H(e_1)$ is in the sub-join, i.e., we need to compute the sub-join
\[R(e_0) \Join R_H(e_1) \Join R_?(e_2) \Join \cdots \Join R_{?}(e_k) \Join \left(\Join_{e \in \bar{\E}} R(e)\right),\]
where each $?$ can be either $H$ or $L$.  The algorithm consists of three steps:
\begin{enumerate}
\item[(2.1)] Compute $R'(e_0) = R(e_0) \ltimes R_H(e_1)$.
\item[(2.2)] Compute $R' = R'(e_0) \Join R_?(e_2) \Join \cdots \Join R_?(e_k) \Join \left(\Join_{e \in \bar{\E}} R(e)\right)$ by any order.
\item[(2.3)] Compute $R_H(e_1) \Join R'$.
\end{enumerate}

We analyze the load in each step: (2.1) is a primitive operation that incurs linear load. To bound the load of (2.2), observe that $|R'| \le {\OUT \over \tau}$, since each tuple in $R'$ joins with at least $\tau$ tuples in $R_H(e_1)$, each producing one final join result.  Thus, the load is bounded by $O(\frac{\IN}{p} + \frac{\OUT}{p \cdot \tau})$. The binary join in (2.3) has input size $\frac{\OUT}{\tau} + \IN$ and output size $\OUT$, incurring a load of $O(\frac{\IN}{p} + \frac{\OUT}{p \cdot \tau} + \sqrt{\frac{\OUT}{p}})$, which dominates the first two steps.  Plugging in the value of $\tau$, the total load is bounded by $O({\IN \over p} + {\sqrt{N_\beta \cdot \OUT} \over p} + \sqrt{\OUT \over p})$, as desired.

\paragraph{Step (3): The sub-join with all $R_L(e_i)$}  
It remains to compute the following sub-join:
\[R(e_0) \Join R_{L}(e_1) \Join \cdots \Join R_{L}(e_k) \Join \left(\Join_{e \in \bar{\E}} R(e)\right).\]

We further divide $R(e_0)$ into heavy and light tuples, as follows.  Let $s = s_1 \cup s_2 \cup \cdots \cup s_k$, and let $v$ be an assignment over attributes $s$.  The set of {\em heavy} assignments in $R(e_0)$ is define as \[H(s, e_0) = \left\{v \in \pi_{s} R(e_0): \prod_{i=1}^k |\sigma_{s_i = \pi_{s_i} v} R_L(e_i)| \ge \tau \right\}.\] Tuples in $R(e_0)$ are classified as {\em heavy} or {\em light}, depending on their projection on attributes $s$, i.e., a tuple $t \in R(e_0)$ is heavy if $\pi_s t \in H(s, e_0)$, and light otherwise.  Similarly, denote the heavy and light tuples in $R(e_0)$ as $R_H(e_0)$ and $R_L(e_0)$, respectively.

These statistics can also be computed using the primitives, but with some more care. For each relation $R_L(e_i)$, we first use sum-by-key to compute $|\sigma_{s_i = v_i} R_L(e_i)|$ for every $v_i \in \pi_{s_i} R_L(e_i)$. This gives us a list of $(v_i, |\sigma_{s_i = v_i} R_L(e_i)|)$ pairs. Then, we use multi-search to find, for each tuple $t \in R(e_0)$, the up to $k$ pairs $(v_i, |\sigma_{s_i = v_i} R_L(e_i)|)$ such that $\sigma_{s_i}t = v_i$.  After this step, each tuple in $R(e_0)$ is attached with $k$ values, and we multiply them together to decide if the tuple is heavy or light.

\paragraph{Step (3.1): The sub-join with $R_H(e_0)$} We first compute the following sub-join:
\[R_H(e_0) \Join R_{L}(e_1) \Join \cdots \Join R_{L}(e_k) \Join \left(\Join_{e \in \bar{\E}} R(e)\right).\]

The algorithm consists of three steps:

\begin{enumerate}
\item[(3.1.1)] Compute $R'(e_0) = R_H(e_0) \Join \left(\Join_{e \in \bar{\E}} R(e)\right)$ by any order.
\item[(3.1.2)] Compute $R'(e_i) = R_H(e_0) \Join R_L(e_i)$ for each $i =1,\cdots,k$.
\item[(3.1.3)] Compute $R'(e_0) \Join R'(e_1) \Join \cdots \Join R'(e_k)$.  Note that each of these relations contains all attributes in $e_0$, so it is a hierarchical join (it is actually tall-flat), so we can use the instance-optimal algorithm in Section~\ref{sec:hierarchical} to compute this join.
\end{enumerate}

Now we analyze the load of each step: First, observe that $|R'(e_0)| \le {\OUT \over \tau}$.  This is because the projection of each tuple in $R'(e_0)$ on $s$ is a heavy assignment, so it will produce at least $\tau$ join results after joining with the $R_L(e_i)$'s.  Therefore, the load of computing the join in (3.1.1) is $O(\frac{\IN}{p} + \frac{\OUT}{p \cdot \tau})$.  Each binary join in (3.1.2) has a load of $O (\frac{\IN}{p} + \sqrt{\frac{\OUT}{p}} )$. Note that each join result $R'(e_i)$ has size bounded by $N_{\beta} \cdot \tau$, since any tuple in $R(e_0)$ can join with at most $\tau$ tuples in $R_L(e_i)$.  Thus, the hierarchical join in (3.1.3) has input size $O(N_{\beta} \cdot \tau + \frac{\OUT}{\tau} )$ and output size $\OUT$, so the instance-optimal algorithm has load $O(\frac{N_\beta \cdot \tau}{p} +\frac{\OUT}{p \cdot \tau} + \sqrt{\frac{\OUT}{p}} )$ according to Corollary~\ref{cor:instace-output}.  All the loads are bounded by $O({\IN \over p} + {\sqrt{N_\beta \cdot \OUT} \over p} + \sqrt{\OUT \over p} )$, as desired.

\paragraph{Step (3.2): The sub-join with $R_L(e_0)$} Finally, we are left with the sub-join 
\[R_L(e_0) \Join R_{L}(e_1) \Join \cdots \Join R_{L}(e_k) \Join \left(\Join_{e \in \bar{\E}} R(e)\right).\]
This is actually the only case where we need recursion:
\begin{enumerate}
\item [(3.2.1)] Compute $R'_L(e_0) = R_L(e_0) \Join R_{L}(e_1) \Join \cdots \Join R_{L}(e_k)$ by any order.
\item[(3.2.2)] If $\bar{\E} \ne \emptyset$,  compute $R'_L(e_0) \Join \left(\Join_{e \in \bar{\E}} R(e)\right)$ recursively.
\end{enumerate}

Now we analyze the load: First, we have $|R'_L(e_0)| \le N_\beta \cdot \tau$, since the projection of each tuple in $R_L(e_0)$ on $s$ is a light assignment.  Thus, the load of step (3.2.1) is $O(\frac{\IN}{p} + \frac{N_\beta \cdot \tau}{p})$, which is also bounded by $O({\IN \over p} + {\sqrt{N_\beta \cdot \OUT} \over p} + \sqrt{\OUT \over p})$.  So far, we have completed the base case of the induction proof.  

For the join to be computed recursively in step (3.2.2), its input size \revise{is} at most $\IN + N_\beta\cdot \tau$ and output size is at most $\OUT$.  More importantly, $N_\beta$ can only become smaller, since $e_0$ becomes a leaf in the residual join and $|R(e_0)|$ is no longer included in $N_\beta$, no matter which node in the residual join is picked to be its new $e_0$.  By the induction hypothesis, computing the residual join recursively incurs a load of $O(\frac{\IN}{p} + \frac{N_\beta \cdot \tau}{p}+ \frac{\sqrt{N_\beta \cdot \OUT}}{p} + \sqrt{\frac{\OUT}{p}} )$, thus bounded by $O({\IN \over p} + {\sqrt{N_\beta \cdot \OUT} \over p} + \sqrt{\OUT \over p})$.

Note that the recursion will increase the constant in the big-Oh, but as the recursion depth depends only on the query not the data size, it does not change the asymptotic result.

This completes the induction proof that the algorithm has a load of $O({\IN \over p} + {\sqrt{N_\beta \cdot \OUT} \over p} + \sqrt{\OUT \over p})$.  Observing that $N_\beta \le \IN$ and $\sqrt{{\OUT \over p}} \le {\sqrt{\IN \cdot \OUT} \over p}$, we obtain the following result.

\begin{theorem}
	\label{the:acyclic}
There is an algorithm that computes any acyclic join with load $O(\frac{\IN}{p} + \frac{\sqrt{\IN \cdot \OUT}}{p} )$ in $O(1)$ rounds.
\end{theorem}

\subsection{Lower bound}
\label{sec:lower-acyclic}

In Section~\ref{lem:3-line-lower} we have constructed a hard instance for the line-3 join and have shown that any algorithm must incur a load of $\Omega(\min\{\frac{\sqrt{\IN \cdot \OUT}}{p \cdot \log \IN}, \frac{\IN}{\sqrt{p}}\})$ on this instance. In this section, we generalize this lower bound to an arbitrary acyclic join that is not r-hierarchical.  Note that for r-hierarchical joins, we can achieve a smaller load $O(\frac{\IN}{p} + \sqrt{\frac{\OUT}{p}})$ (see Corollary~\ref{cor:instace-output}), so this establishes a separation between r-hierarchical joins and acyclic joins. 

The basic idea in the lower bound is that any acyclic join must ``include'' a line-3 join, such that any algorithm computing the acyclic join must also compute the line-3 join.  This is more formally captured by the following structural lemma on acyclic and r-hierarchical joins.  To state the lemma, we need some terminology.  In a hypergraph $\Q=(\V, \E)$, a {\em path} between $x, y \in \V$, denoted $P(x,y)$, is a sequence of vertices starting with $x$ and ending with $y$, such that each consecutive pair of vertices appear together in an edge. The length of a path is defined as the number of vertices in $P(x,y)$ minus 1.  A path $P(x, y)$ is {\em minimal} if there is no other path $P'(x, y)$ that is a strict subsequence of $P(x, y)$. It is easy to see that $P(x_1, x_k) = (x_1, x_2, \cdots,x_k)$ is minimal if and only if there exists no edge $e \in \E$ containing $x_i$ and $x_j$ with $|j-i| > 1$. Note that a shortest path must be minimal, but not vice versa.

\begin{lemma}
	\label{lem:non-hierarchical}
	An acyclic join is not r-hierarchical if and only if it has a minimal path of length $3$.
\end{lemma}

The proof is given in Appendix~\ref{sec:proof-lemma-path3}. With this lemma, we can extend the hard instance for the line-3 join to any acyclic but non-r-hierarchical join $\Q=(\V, \E)$. Let $(x_1,x_2,x_3,x_4)$ be a minimal path of length 3 in $\Q$, and suppose $\{x_1, x_2\} \subseteq e_1, \{x_2, x_3\} \subseteq e_2, \{x_3, x_4\} \subseteq e_3$.  Let $\R = \{ R_1(x_1,x_2),$ $R_2(x_2,x_3), R_3(x_3,x_4)\}$ be the hard instance for the line-3 join.  We construct the hard instance $\R' = \{R'(e) : e\in \E\}$ for $\Q$ as follows.  The domain of $x_i, i=1,2,3,4$ is the same as in $\R$.  For any other attribute $y$, set $|\dom(y)|= 1$.

Since the path is minimal, each $e\in \E$ must fall into one of the following three cases:

\begin{enumerate}
\item For any $e$ with $e \cap \{x_1, x_2, x_3, x_4\} = \emptyset$, $R'(e)$ contains only one tuple connecting the only value in the domains of attributes in $e$.
\item  If $e \cap \{x_1, x_2, x_3, x_4\} = \{x_i\}, i=1,2,3,4$, then $R'(e)$ contains $|\dom(x_i)|$ tuples, each having a distinct value of $\dom(x_i)$. 
\item  If $e \cap \{x_1, x_2, x_3, x_4\} = \{x_i, x_{i+1}\}, i=1,2,3$, then $R'(e)$ contains $|R_i(x_i, x_{i+1})|$ tuples such that $\pi_{x_i, x_{i+1}} R'(e) = R_i(x_i, x_{i+1})$.  
\end{enumerate}

It can be easily verified that $\Q(\R')$ is exactly the join results of the line-3 join on $\R$, so the same lower bound applies.  However, since the output size of the line-3 join is at most $\IN^2$, we do have a condition on $\OUT$:

\begin{theorem}
	\label{thm:acyclic-lower}
	For an acyclic but non-r-hierarchical join and any $\IN \ge p^{3/2}, \OUT \le \IN^2$, there exists an instance $\R$ with input size $\Theta(\IN)$ and output size $\Theta(\OUT)$ such that any tuple-based algorithm computing it in $O(1)$ rounds must have a load of $\Omega(\min\{\frac{\sqrt{\IN \cdot \OUT}}{p \cdot \log \IN}, \frac{\IN}{\sqrt{p}}\})$.
\end{theorem}

Similar to the line-3 join, this lower bound shows that our acyclic join algorithm is output-optimal (up to a logarithmic factor) when $\OUT \le p\cdot \IN$.

Furthermore, the same argument for Corollary~\ref{cor:instance-impossible} can be used here to show that instance-optimal algorithms do not exist for any acyclic but non-r-hierarchical join.

\begin{corollary}
  \label{cor:instance-impossible2}
  For any $\IN \ge p^{3/2}$, there is an instance $\R$ with input size $\Theta(\IN)$ for any acyclic but non-r-hierarchical join, such that any tuple-based algorithm that computes the join in $O(1)$ rounds must have a load of $\Omega({\IN \over p^{1/2} \log\IN})$, while $\Li(p, \R)=O({\IN \over p})$.
\end{corollary}


%% file: aggregate.tex
\section{Join-Aggregate Queries}
\label{sec:join-aggregate}

We consider join-aggregate queries over {\em annotated relations} \cite{green07:_proven,joglekar16:_ajar}.  Let $(\mathbb{R}, \oplus, \otimes)$ be a commutative semiring.  Every tuple $t$ is associated with an {\em annotation} $w(t) \in \mathbb{R}$.  Let $\Q=(\V, \E)$ be a join hypergraph.  The annotation of a join result $t\in \Q(\R)$ is $w(t) := \otimes_{t_e \in R(e), \pi_e t = t_e, e\in \E} w(t_e)$. Let $\bm{y} \subseteq \V$ be a set of {\em output attributes} and $\bm{\bar{y}} = \V - \bm{y}$ the non-output attributes.  A {\em join-aggregate query} $\Q_{\bm{y}}(\R)$ asks us to compute 
\[\oplus_{\bm{\bar{y}}} \Q(\R) = \oplus_{\bm{\bar{y}}} \Q(\R) =  \left\{ (t_{\bm{y}}, w(t_{\bm{y}})) : t_{\bm{y}} \in \pi_{\bm{y}} \Q(\R), w(t_{\bm{y}}) = \oplus_{t \in \Q(\R): \pi_{\bm{y}} t = t_{\bm{y}}} w(t)\right\}.\] In plain language, a join-aggregate query first computes the join $\Q(\R)$ and the annotation of each join result, which is the $\otimes$-aggregate of the tuples comprising the join result.  Then it partitions $\Q(\R)$ into groups by their projection on $\bm{y}$.  Finally, for each group, it computes the $\oplus$-aggregate of the annotations of the join results. 

Many queries can be formulated as special join-aggregate queries.  For example, if we take $\mathbb{R}$ to be the domain of integers, $\oplus$ to be addition, $\otimes$ to be multiplication, and set $w(t)=1$ for all $t$, then it becomes the {\tt COUNT(*) GROUP BY} $\bm{y}$ query; in particular, if $\bm{y}=\emptyset$, the query computes $|\Q(\R)|$.

The join-project query $\pi_{\bm{y}}\Q(\R)$, also known as a {\em conjunctive query}, is a special join-aggregate query, and we extend the terminology from \cite{bagan2007acyclic} to join-aggregate queries.  A {\em width-1 GHD} of a hypergraph $\Q = (\V,\E)$ is a tree $\T$, where each node $u \in \T$ is a subset of $\V$, such that
\begin{enumerate}
\item (coherence) for each attribute $x \in \V$, the nodes containing $x$ are connected in $\T$; 
\item (edge coverage) for each hyperedge $e \in \E$, there exists a node $u \in \T$ such that $e \subseteq u$; and
\item (width-1) for each node $u \in \T$, there exists a hyperedge $e \in \E$ such that $u \subseteq e$.
\end{enumerate}
Given a set of output attributes $\bm{y}$ (a.k.a. {\em free variables}), we say that $\T$ is {\em free-connex} if there is a subset of connected nodes of $\T$ including its root, denoted as $\T'$ (such a $\T'$ is said to be a {\em connex} subset), such that $\bm{y} = \bigcup_{u \in \T'} u$.  A join-aggregate query $\Q_{\bm{y}}(\R)$ is {\em free-connex} if it has a free-connex width-1 GHD.

As preprocessing, we remove the dangling tuples and then apply the reduce procedure repeatedly to remove an $e \in \E$ if there is another $e' \in \E$ such that $e \subset e'$.  Note that while dangling tuples can be just discarded, we cannot simply discard $R(e)$ in the reduce procedure.  To ensure that the annotations will be computed correctly, we should replace $R(e')$ with $R(e) \Join R(e')$ and then discard $R(e)$.  Note that by the earlier definition, the annotation of a join result is the $\otimes$-aggregate of the annotations of tuples comprising the join result, so the annotation in $R(e)$ are aggregated into those in $R(e')$.

We find a free-connex width-1 GHD $\T$ of $\Q$ \cite{bagan2007acyclic,bagan2009algorithmes}.  Note that the nodes of $\T$ also define a hypergraph, and can be regarded as another join-aggregate query, but with the property that it has a free-connex subset $\T'$ such that $\bm{y} = \bigcup_{u \in \T'} u$.  We construct an instance $\R_{\T}= \{ R(u) : u \in \T\}$ such that $\Q_{\bm{y}}(\R) = \T(\R_\T)$, where $\T(\R_\T)$ denotes the result of running the query defined by $\pi_{\bm{y}} \T$ on $\R_\T$.  Observe that on a reduced $\Q$, the condition $e \subseteq u$ in property (2) of a width-1 GHD can be replaced by $e = u$, since if $e \subset u$ and $u \subseteq e'$ for some other $e' \in \E$ due to property (3), we would find $e \subset e'$.  This implies that $\T$ has only two types of nodes: (1) all hyperedges in $\E$, and (2) nodes that are a proper subset of some $e\in \E$.  Then we construct $\R_\T$ as follows.  For each $u \in \T$ of type (1), we set $R(u) := R(e)$ where $e=u$; for each $u \in \T$ of type (2), we set $R(u) := R(e)$ for any $e\in \E, u \subset e$, but the annotations of all tuples in $R(u)$ are set to $1$ (the $\otimes$-identity).  Below, we will focus on computing $\T(\R_\T)$.

Joglekar et al. \cite{joglekar16:_ajar} modified the Yannakakis algorithm into {\sc AggroYannakakis}, and showed that it has load $O(\frac{\IN}{p} + \frac{\OUT}{p})$ on any free-connex join-aggregate query\footnote{The bound stated in \cite{joglekar16:_ajar} is actually $O({(\IN + \OUT)^2 \over p})$, because they used a sub-optimal binary join algorithm as the subroutine following \cite{afrati2014gym}.  Replacing it with the optimal binary join algorithm in \cite{beame14:_skew,hu17:_output} yields the claimed bound.  In addition, they only considered {\em simple} join-aggregate queries, which are a strict subclass of free-connex queries. But after our conversion from $\Q_{\bm{y}}(\R)$ to $\T(\R_\T)$, their algorithm actually works for all free-connex queries.}.  Since we want to avoid the sub-optimal $O ({\OUT \over p})$ term, we modify their algorithm into {\sc LinearAggroYannakakis}, which runs with linear load.  It aggregates over all the non-output attributes, returning a modified query $\T'(\R_{\T'})$ that only has the output attributes.  The guarantees of {\sc LinearAggroYannakakis} is stated in the following lemma.

\begin{algorithm}[t]
	\caption{{\sc LinearAggroYannakakis}$(\T, \T', \R_{\T})$}
	\label{alg:aggregate}
	$\R_{\T'} = \emptyset$; \\
	\For{$u \in \T$ in some bottom-up ordering}{
		\uIf {$u \in \T'$} {
			add $R(u)$ to $\R_{\T'}$;\\ 
		}
		\Else{
			$\bm{\bar{y}}' \gets \bm{\bar{y}} \cap \{x \in \V: TOP_{\T}(x) = u\}$ \;
			$R(u) \gets \oplus_{\bm{\bar{y}}'} R(u)$ (sum-by-key)\\      
			$u' \gets$ parent of $u$ in $\T$ \;
			$R(u') \gets R(u') \Join R(u)$; (multi-search)\\
		}
	}
\end{algorithm}

\begin{lemma}
  \label{lem:linearaggro}
  {\sc LinearAggroYannakakis} is a constant-round, linear-load algorithm that, given any free-connex width-1 GHD $\T$ and an instance $\R_\T$, returns an instance $\R_{\T'}$ such that $\T(\R_\T) = \T'(\R_{\T'})$, where $\T'$ is the free-connex subset of $\T$.
\end{lemma}

\begin{proof}
	Let $\T$ be a width-1 free-connex GHD and $\T'$ be the connex subset of $\T$ such that $\bm{y} = \bigcup_{u \in \T'} u$.  For an attribute $x$, denote the highest node in $\T$ containing $x$ as $TOP_{\T}(x)$.  Below, we describe {\sc LinearAggroYannkakakis}, an algorithm that converts $\R_\T$ into $\R_{\T'}$ such that $\T(\R_\T) = \T'(\R_{\T'})$.
	
	The {\sc LinearAggroYannkakakis} algorithm 
	visits each node $u \in \T$ in a bottom-up fashion over $\T$.  If $u\in \T'$, i.e., all its attributes are output attributes, we add $R(u)$ to $\R_{\T'}$ (line 4).  Otherwise, we aggregate over $\bm{\bar{y}'}$, which are the non-output attributes in $u$ that do not appear in the ancestors of $u$ (line 6--7).  This is a sum-by-key problem.  Note that after the aggregation, the attributes of $R(u)$ are $u - \bm{\bar{y}}'$.  Let $u'$ be the parent of $u$ in $\T$. Note that $u'$ always exists since the root of $\T$ must be in $\T'$. Then we replace $R(u')$ by $R(u') \Join R(u)$ (line 9).  Below we show how this join can be done in linear load.  Consider any non-output attribute $x \in u-\bm{\bar{y}}' - \bm{y}$. Since $TOP_{\T}(x)$ is an ancestor of $u$, we have $x \in u'$.  Consider any output attribute $y \in u \cap \bm{y}$. In the connex subset $\T'$, there exists $u'' \in \T$ such that $y \in u''$. Each node on the path from $u''$ to $u$ must contain attribute $y$, including $u'$. Thus, we must have $u-\bm{\bar{y}}' \subseteq u'$. This means that tuples in $R(u') \Join R(u)$ are actually the same as those in $R(u')$, except that we update the annotation of each $t\in R(u')$ as $w(t) \gets w(t) \otimes w(t')$, where $t' \in R(u), t' = \pi_{u-\bm{\bar{y}}'}t$.  Thus, this can be done by the multi-search primitive in linear load. Because this algorithm never increases the size of any relation, the two primitive operations (line 7 and 9) incur linear load throughout the bottom-up traversal of $\T$.
	
	It should be obvious from the algorithm description above that  {\sc LinearAggroYannkakakis} incurs linear load, but we still need to argue for its correctness. Note that $\R_{\T'}$ has only output attributes. It suffices to show that $\T(\R_\T) = \T'(\R_{\T'})$.
	
	Joglekar et al. \cite{joglekar16:_ajar} have shown that for any leaf $u \in \T$ and its parent $u'$, performing the operation in lines 6--9 and then discarding $R(u)$ does not change the query results.  {\sc AggroYannkakakis} performs this operation over all the relations of $\T$ in a bottom-up fashion, and applying this fact inductively means that the root relation becomes the final query result in the end, but this incurs load $O({\IN \over p} + {\OUT \over p})$. {\sc LinearAggroYannkakakis} performs this operation on a subset of relations, and stops as soon as it sees a node in $\T'$. Then applying the result of \cite{joglekar16:_ajar} inductively up until $\T'$ proves our claim.
\end{proof}

Because $\T'$ is acyclic, we can run our output-optimal algorithm to compute $\T'(\R_{\T'})$. More precisely, when the algorithm emits a join result, we compute the $\otimes$-aggregate of the tuples comprising the join result.  Note that in the following result,  $\OUT= |\Q_{\bm{y}}(\R)|$, i.e., the size of the final output, which can be much smaller than $|\Q(\R)|$.

\begin{theorem}
	\label{the:join-aggregate}
	There is an algorithm that computes any free-connex join-aggregate query in $O(1)$ rounds with load $O(\frac{\IN}{p} + \frac{\sqrt{\IN \cdot \OUT}}{p})$.
\end{theorem}

Observing that the join size of a (non-aggregate) join is a special join-aggregate query with $\bm{y}=\emptyset$, we obtain the following result, which has been used as a primitive. Note that there is no circular dependency here, because it only uses {\sc LinearAggroYannakakis}.

\begin{corollary}
	For any acyclic join $\Q$ and any instance $\R$, $|\Q(\R)|$ can be computed in $O(1)$ rounds with load $O (\frac{\IN}{p})$.
\end{corollary}

Furthermore, if $\T'$ is r-hierarchical, we run our instance-optimal algorithm to compute $\T'(\R_{\T'})$.  In fact, we can precisely characterize the class of queries with an r-hierarchical $\T'$. A query is called {\em out-hierarchical} if it is free-connex and its residual query by removing all non-output attributes is r-hierarchical.  

\begin{lemma}
  \label{lem:out-hierarchical}
  A join-aggregate query $\Q_{\bm{y}}$ is out-hierarchical if and only if it has a width-1 GHD $\T$ with a connex subset $\T'$ such that $\bm{y} = \bigcup_{u \in \T'} u$ and $\T'$ is r-hierarchical.
\end{lemma}

\begin{proof}
	First we have known that join-aggregate query $\Q_{\bm{y}}$ is free-connex iff it has a width-1 GHD $\T$ with a connex subset $\T'$ such that $\bm{y} = \bigcup_{u \in \T'} u$. Consider $\Q_{out} = (\bm{y}, \{e \cap \bm{y}: e \in \E\})$ the residual query of $\Q_{\bm{y}}$ after removing all non-output attributes. Then it suffices to show that for a free-connex query $\Q_{\bm{y}}$, $\Q_{out}$ is r-hierarchical iff $\T'$ is r-hierarchical. 
	
	An edge $e \in \E$ is {\em out-irreducible} if there exists no $e' \in \E$ such that $e \cap \bm{y} \subset e' \cap \bm{y}$ or $e \subset e'$; otherwise {\em out-reducible}. We first claim that for each out-irreducible $e \in \E$ there exists one node $v \in \T'$ such that $e \cap \bm{y} \subseteq v$. Consider the node $u \in \T$ such that $u = e$. If $u \in \T'$, the claim holds trivially. Otherwise, consider the lowest ancestor of $u$ in $\T'$ as $v$. As each output attribute $x \in u \cap \bm{y}$ appears in some node of $\T'$, it also appears in $v$ due to the coherence constraint. Thus, $e \cap \bm{y} \subseteq v$.
	
	Recall that for each node $u \in \T$, there exists an edge $e \in \E$ such that $u \subseteq e$. Correspondingly, for each node $v \in \T'$, there exists an edge $e \in \E$ such that $v \subseteq e \cap \bm{y}$.  Thus, for each out-irreducible $e \in \E$, there exists one node $v \in \T'$ such that $e \cap \bm{y} = v$, since if $e \cap \bm{y} \subset v$ and $v \subseteq e' \cap \bm{y}$ for some other $e' \subseteq \E$, $e$ would be out-reducible. This implies that $\T'$ has only two types of nodes: (1) $e \cap \bm{y}$ for each out-irreducible $e \in \E$, and (2) a proper subset of $e \cap \bm{y}$ for some $e \in \E$. 
	
	Not surprisingly, $\Q_{out}$ also have two types of edges, (1) $e \cap \bm{y}$ for each out-irreducible $e \in \E$, and (2) $e \cap \bm{y}$ for each out-reducible $e \in \E$. Nodes in $\T'$ of type (1) are one-to-one mappings to edges in $\Q_{out}$ of type (1). Moreover, after applying the reduce procedure repeatedly on $\T'$ or $\Q_{out}$, only nodes or edges of type (1) can survive. Thus, the reduced query of $\Q_{out}$ is hierarchical iff the reduced query of $\T'$ is hierarchical, and the $\Q_{out}$ is r-hierarchical iff $\T'$ is r-hierarchical.
\end{proof}

\begin{theorem}
	\label{the:out-hierarchical}
	For out-hierarchical query $\Q_{\bm{y}}$ and any instance $\R$, there is an algorithm computing it in $O(1)$ rounds with load $O(\frac{\IN}{p} + \Li(p, \R, \bm{y}))$.
\end{theorem}

Note that the instance-optimal lower bound $\Li$ for a join-aggregate query is defined with respect to the output attributes only, i.e.,
\[\Li(p, \R, \bm{y}) := \max_{S \subseteq \E} \left(\frac{|\pi_{\bm{y}}\Q(\R, S)|}{p}\right)^{\frac{1}{|S|}},\]
 where $\pi_{\bm{y}}\Q(\R, S) = \pi_{\bm{y}} ((\Join_{e \in S} R(e)) \ltimes \Q(\R))$.


%% file: triangle.tex
\section{A Lower Bound on Triangle Join}
\label{sec:triangle}

Finally, we look beyond acyclic joins.  In particular, we give an output-sensitive lower bound on the triangle join $\Q_\triangle = R_1(B, C) \Join R_2(A, C) \Join R_3(A,B)$.  For $\Q_\triangle$, a worst-case lower bound of $\Omega ({\IN \over p^{2/3}} )$ is known, by the following argument: A server loading $L$ tuples can emit at most $O(L^{3/2})$ join results by the AGM bound \cite{atserias2008size}, while the join size of $\Q_\triangle$ can be as large as $\Omega(\IN^{3/2})$.  Then setting $p\cdot L^{3/2} = \Omega(\IN^{3/2})$ yields this lower bound.  However, if $\OUT$ is used as a parameter, this argument only leads to a lower bound of $\Omega ( ({\OUT \over p} )^{2/3} )$.  Below, we improve this lower bound to the following:

 \begin{theorem}
   \label{the:triangle}
   For any $\IN/\log^2 \IN \ge 3 p^3, \OUT$, there exists an instance $\R$ for $\Q_\triangle$ with input size $\Theta(\IN)$ and output size $\Theta(\OUT)$ such that any tuple-based algorithm computing it in $O(1)$ rounds must have a load of $\Omega (\min   \{{\IN \over p} + \frac{\OUT}{p \log N},\frac{\IN}{p^{2/3}}   \} )$.
 \end{theorem}


 \begin{proof}
 	When $\OUT \le \IN$, the claimed lower bound simplifies to $\Omega({\IN \over p})$, so we will only consider the case $\OUT > \IN$.  Let $N = \IN /3$ and $\tau = \frac{\OUT}{N}$.  Note that $\tau \le \sqrt{N}$ as implied by the AGM bound. Our construction of the hard instance $\R$ is illustrated in Figure~\ref{fig:triangle}.
 	\begin{figure}[h]
 		\centering
 		\includegraphics[scale=1.0]{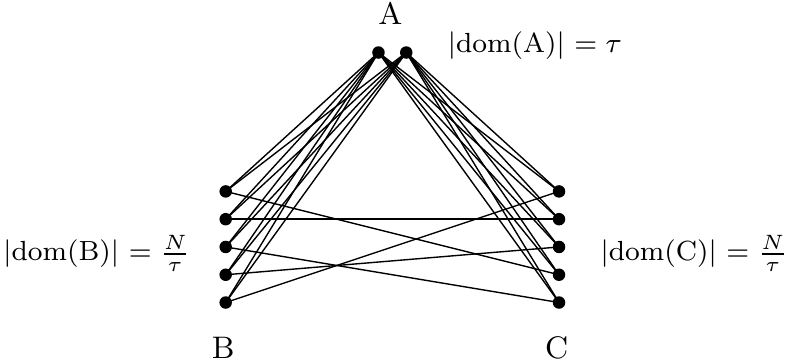}
 		\caption{A randomly constructed hard instance.}
 		\label{fig:triangle}
 	\end{figure}
 	
 	Set $|\dom(A)| = \tau$, and $|\dom(B)|= |\dom(C)| = \frac{N}{\tau}$.  Set $R_2(A, B) = \dom(A) \times \dom(C)$ and $R_3(A,B) = \dom(A) \times \dom(B)$.  The relation $R_1(B,C)$ is constructed randomly, in which each distinct value in $B$ and each distinct value of $C$ have a probability of $\frac{\tau^2}{N}$ to form a tuple. Note that relations $R_2$ and $R_3$ are deterministic and always have $N$ tuples. Relation $R_1$ is probabilistic with $N$ tuples in expectation. So this instance has input size $3N = \IN$ and output size $({N \over \tau})^2 \cdot {\tau^2 \over N} \cdot \tau = \OUT$ in expectation. By the Chernoff bound, the probability that the input size and output size deviate from their expectation by more than a constant factor is at most $\exp(-\Omega(|R_1|)) = \exp(-\Omega(N))$.
 	
 	Similar to the proof of Lemma \ref{lem:3-line-lower}, we will show that with positive probability, an instance constructed this way will have a bounded $J(L)$, the maximum number of join results a server can produce, if it loads at most $L$ tuples from each relation.  Then setting $p\cdot J(L) = \Omega(\OUT)$ yields a lower bound on $L$.
 	
 	To bound $J(L)$, we first argue that on any instance constructed as above, we can limit the choice of the $L$ tuples loaded from $R_2(A,C)$ ($R_3(A,B)$, respectively) by any server to the form $X \times Y$ for some $X\in \dom(A), Y\in \dom(C)$ ($Y \in \dom(B)$, respectively), i.e., the algorithm should load tuples from $R_2(A,B)$ and $R_3(A,C)$ in the form of a Cartesian product.  More precisely, we show below that making such a restriction will not make $J(L)$ smaller by more than a constant factor.  
 	
 	Suppose a server has loaded $L$ tuples from $R_1(A,B)$.  Then the server needs to decide which $L$ tuples from $R_2$ and $R_3$ to load to maximize the number of triangles formed.  This is a combinatorial optimization problem that can be formulated as an integer linear program (ILP). Introduce a variable $u_{ab}$ for each pair $a \in \dom(A), b \in \dom(B)$ and a variable $v_{ac}$ for each pair $a \in \dom(A), c \in \dom(C)$.  Also let $I(bc) = 1$ if $(b,c) \in R_1$ is loaded by the server, and $0$ otherwise.  Then $ILP_1$ below defines this optimization problem, where $a$ always ranges over $\dom(A)$, $b$ over $\dom(B)$, $c$ over $\dom(C)$ unless specified otherwise.
 	\[ \begin{array}{ccll}
 	ILP_1:& maximize & \sum_{a, b, c} I(bc) \cdot u_{ab} \cdot v_{ac} &  \\
 	&s.t. & \sum_{a, b} u_{ab} \le L &\\
 	& & \sum_{a, c} v_{ac} \le L & \\
 	& & u_{ab} \in \{0,1\}, \ v_{ac} \in \{0,1\}, \forall a, b, c \\
 	\\
 	ILP_2:& maximize & \sum_{b, c} I(bc) \cdot u_{ab} \cdot v_{ac} &  \\
 	&s.t. & \max\{\sum_{b} u_{ab}, \sum_{c} v_{ac}\} \le w &\\
 	& & u_{ab} \in \{0,1\}, \ v_{ac} \in \{0,1\}, \forall b, c\\
 	\\
 	ILP_3:& maximize & \sum_{a} \Delta(w_a) & \\
 	& s.t. & \sum_{a} w_a \le 2L, \ w_{a} \in \{0,1, \cdots, L\}, \forall a 
 	\end{array}\]
 	
 	We transform $ILP_1$ into another ILP, shown as $ILP_3$ above. $ILP_3$ uses a function $\Delta(w)$, which denotes the optimal solution of $ILP_2$.  $ILP_2$ is parameterized by $w$ and $a$, which finds the maximum number of triangles that can be formed with the tuples loaded from $R_1(B,C)$ and $a\in \dom(A)$, subject to the constraint that at most $w$ tuples containing $a$ are loaded from $R_2$ and $R_3$.  Because all values $a\in \dom(A)$ are structurally equivalent, the optimal solution of $IL_2$ does not depend on the particular choice of $a$, which is why we write the optimal solution of $ILP_2$ as $\Delta(w)$.  Also, it is obvious that $\Delta(.)$ is a non-decreasing function. Then, $ILP_3$ tries to find the optimal allocation of the $L$ tuples to different values $a\in \dom(A)$ so as to maximize the total number of triangles formed.   Let the optimal solutions of $ILP_1, ILP_3$ be $OPT_1, OPT_3$, respectively.  Because $ILP_3$ only restricts the server to load at most $2L$ tuples from $R_2$ and $R_3$ in total, any feasible solution to $ILP_1$ is also a feasible solution to $ILP_3$, so $OPT_1 \le OPT_3$. Next we construct a feasible solution of $ILP_3$ with the Cartesian product restriction above, and show that it is within a constant factor from $OPT_3$, hence $OPT_1$. 
 	
 	Let $w^* = \arg \max_{\frac{L}{\tau} \le w \le L} \frac{L}{w} \cdot \Delta(w)$. We choose $\frac{L}{w^*}$ values arbitrarily from $\dom(A)$ and allocate $w^*$ tuples to each such $a$.  For each such $a$, we use the optimal solution of $ILP_2$ to find the $w^*$ tuples to load from $R_2$ and $R_3$.  Note that the optimal solution is the same for every $a$, so each $a$ will choose the same sets of $b$'s and $c$'s.  Thus, this feasible solution loads tuples from $R_1$ and $R_2$ in the form of Cartesian products.  The number of triangles formed is $W = \frac{L}{w^*} \cdot \Delta(w^*)$.  We show that this is a constant-factor approximation of $OPT_3$.
 	\begin{lemma}
 		\label{lem:triangle}
 		$W \ge \frac{1}{3} OPT_3 \ge \frac{1}{3} OPT_1$.
 	\end{lemma}
 	
 	\begin{proof}
 		Suppose $OPT_3$ chooses a set of values $A^*$ from $A$, and each $a \in A^*$ has $w_a$ tuples loaded from $R_2$ and $R_3$. A value $a \in A^*$ is {\em efficient} if $\frac{\Delta(w_a)}{w_a} \ge \frac{ \Delta(w^*)}{w^*}$, otherwise {\em inefficient}.  Denote the set of efficient values as $A^*_1$ and inefficient values as $A^*_2$. Note that for every efficient value $a$, $w_a \le \frac{L}{\tau}$ by the definition of $w^*$.
 		
 		We relate $W$ and $OPT_3$ by showing how to cover all the triangles reported by $OPT_3$ with the feasible solution constructed above.  First, we use $\frac{\sum_{a \in A^*_2} w_a}{3w^*} $ values of $A$ each with $w^*$ tuples from $R_2$ and $R_3$ to cover the triangles reported by $A^*_2$. The total number of tuples needed is at most $\frac{2}{3} \sum_{a \in A^*_2} w_a \le \frac{4}{3}L$. The number of triangles that can be reported is
 		\[
 		\frac{\sum_{a \in A^*_2} w_a}{3w^*}  \cdot \Delta(w^*)  \ge \frac{1}{3} \sum_{a \in A^*_2} w_a \cdot \frac{\Delta(w_a)}{w_a} = \frac{1}{3} \sum_{a \in A^*_2} \Delta(w_a). 
 		\]
 		
 		Next, we use $\frac{L}{3w^*}$ values each with $w^*$ tuples from $R_2$ and $R_3$ to cover the triangles reported by $A^*_1$. The total number of tuples needed is $\frac{2}{3}L$. Recall that $w_a \le \frac{L}{\tau}$ for each $a \in A^*_1$. The number of triangles that can be reported is  
 		\begin{align*}
 		\frac{L}{3w^*} \cdot \Delta(w^*) \ge \frac{L}{3} \cdot \frac{\Delta(\frac{L}{\tau})}{\frac{L}{\tau}} = \frac{\tau}{3} \cdot \Delta(\frac{L}{\tau})\ge\frac{1}{3} \sum_{a \in A^*_1} \Delta(w_a),
 		\end{align*}
 		where the rationale behinds the last inequality is that there are at most $\tau$ values in $A^*_1$ and there is $\Delta\left(\frac{L}{\tau}\right) \ge \Delta(w_a)$ for each $a \in A^*_1$ by the non-decreasing property of $\Delta(.)$.
 		
 		Combining the two parts for the optimal solution $A^*$, our alternative solution loads at most $2L$ tuples from $R_2$ and $R_3$, and can report at least $\frac{1}{3} \cdot OPT_3$ triangles.
 	\end{proof}

 	
 	Next, we show that with positive probability (actually high probability), we obtain an instance on which $J(L)$ is bounded.  By the analysis above, we only need to consider the case where tuples from $R_2$ and $R_3$ are loaded in the form of Cartesian products. One value $b \in \dom(B)$ is loaded if at least one tuple $t \in R_3(A,B)$ with $\pi_B t = b$ is loaded. Similarly, value $c \in \dom(C)$ is loaded if at least one tuple $t \in R_2(A,C)$ with $\pi_C t = c$ is loaded.  Suppose $\alpha$ and $\beta$ distinct values from $B$ and $C$ are loaded respectively. Note that we must have $1 \le \alpha, \beta \le \min\{L, \frac{N}{\tau}\}$. Without loss of generality, assume $\alpha \le \beta$. Due to Cartesian product constraint, the number of distinct values loaded from $A$ is at most $\tau = \min\{\frac{L}{\beta}, \tau\}$.
 	
 	\paragraph{Case 1: $\bm{\alpha \beta \le \frac{NL}{\tau^2}}$} 
 	
 	We first upper bound the probability that the server can report many triangles on a random instance, for a particular choice of $\alpha$ values loaded from $\dom(B)$ and $\beta$ values from $\dom(C)$.  Since at most $\gamma$ distinct values from $A$ are loaded, each tuple loaded from $R_1(B,C)$ can form at most $\gamma$ triangles. Because each $(b,c)$ pair have probability ${\tau^2 \over N}$ to form a tuple in $R_1(B,C)$, on a random instance, we expect to see $\frac{\tau^2\alpha \beta}{N}$ tuples and $\frac{\tau^2\alpha \beta \gamma}{N}$ triangles. Note that this is always smaller than $\tau \sqrt{\frac{L^3}{N}}$: (1) If $\tau \beta \ge L$, $\gamma = \frac{L}{\beta}$, then $\frac{\tau^2\alpha \beta \gamma}{N} \le \frac{\tau^2L}{N} \cdot \sqrt{\alpha \beta} \le \frac{\tau^2L}{N} \cdot \sqrt{\frac{NL}{\tau^2}} \le \tau \sqrt{\frac{L^3}{N}}$; (2)  Otherwise, $\gamma = \tau$, then
 	$\frac{\tau^2\alpha \beta \gamma}{N} \le \frac{\tau^3 \beta^2}{N} \le  \frac{\tau^3}{N} \cdot \frac{L^2}{\tau^2} \le \frac{\tau L^2}{N} \le \tau \sqrt{\frac{L^3}{N}}$. This server can report more than $\delta \tau \sqrt{\frac{L^3}{N}}$ triangles, for some $\delta > 1$,  if more than $\frac{\delta \tau}{\gamma} \sqrt{\frac{L^3}{N}}$ tuples exist among those $\alpha \beta$ pairs. By Chernoff bound, this happens with probability no larger than $\exp\left(-\Omega(\frac{\delta  \tau}{\gamma} \sqrt{\frac{L^3}{N}})\right)$.
 	
 	This is the probability that the server succeeds in reporting many triangles under a particular choice of $\alpha$ values loaded from $\dom(B)$ and $\beta$ values from $\dom(C)$.  There are $O\left((\frac{N}{\tau})^2\right)$ possible $(\alpha ,\beta )$ pairs.  For each $(\alpha ,\beta)$ pair, there are $O\left(\tau^\gamma\right)$ choices of loading $\gamma$ values from $A$, $O\left((\frac{N}{\tau})^\alpha \right)$ choices of loading $\alpha $ values from $B$, and $O\left((\frac{N}{\tau})^\beta \right)$ choices from $C$.  Thus the server has $\exp(O((\alpha + \beta + \gamma)\log N))$ possible choices.  By the union bound, the probability that any of these choices produces more than $\delta \tau \sqrt{\frac{L^3}{N}}$ join results is at most
 	\begin{equation}
 	\label{eq:1}
 	\exp\left(-\Omega\left(\frac{\delta \tau}{\gamma} \sqrt{\frac{L^3}{N}}\right) + O((\alpha  + \beta + \gamma)\log N)\right), 
 	\end{equation}
 	which is exponentially small if
 	\[\frac{\delta \tau}{\gamma} \sqrt{\frac{L^3}{N}} \ge c_1 \cdot \frac{L}{\gamma} \log N \ge c_1 \cdot \beta \log N\]
 	and
 	\[\frac{\delta \tau}{\gamma} \sqrt{\frac{L^3}{N}} \ge c_2 \cdot \frac{\tau^2}{\gamma} \log N  \ge c_2 \cdot \gamma \log N \] for some sufficiently large constants $c_1, c_2$. Rearranging, this becomes
 	\begin{equation}
 	\label{eq:6}
 	\delta \ge c_3 \cdot \max \left \{\frac{1}{\tau} \sqrt{\frac{N}{L}} \log N, \tau \sqrt{\frac{N}{L^3}} \log N\right \}
 	\end{equation}
 	for some sufficiently large constant $c_3$. Under this condition, the probability in $(\ref{eq:1})$ is at most $\exp\left(-\Omega\left(\delta \sqrt{\frac{L^3}{N}}\right)\right)$.
 	
 	\paragraph{Case 2: $\bm{\alpha \beta  > \frac{NL}{\tau^2}}$}
 	In this case, we have $\beta \ge \frac{\sqrt{NL}}{\tau}$.  The server loads $\frac{L}{\beta }$ distinct values from $A$, so each tuple loaded from $R_1$ can form at most $\frac{L}{\beta }$ triangles.  The server can load at most $L$ tuples from $R_1$, so at most $\frac{L^2}{\beta } \le \delta \tau \sqrt{\frac{L^3}{N}}$ triangles can be reported, for any
 	\begin{equation}
 	\label{eq:7}
 	\delta \ge 1.  
 	\end{equation}

 	\bigskip Combining these two cases, under the condition (\ref{eq:6}) and (\ref{eq:7}) on $\delta$, with high probability the server cannot find any way to load $L$ tuples to report more than $\delta \tau \sqrt{\frac{L^3}{N}}$ triangles. 
 	Therefore, on these instances, we have
 	\begin{equation}
 	\label{eq:9}
 	J(L) \le \delta \tau \sqrt{\frac{L^3}{N}},
 	\end{equation} where we set
 	\begin{equation}
 	\label{eq:8}
 	\delta = c_3 \max\left\{\frac{1}{\tau} \sqrt{\frac{N}{L}} \log N, \tau \sqrt{\frac{N}{L^3}} \log N, 1\right\}.  
 	\end{equation}
 	With the facts that $\frac{N}{p} \le L$ and $\OUT \le N^{\frac{3}{2}}$,
 	we observe
 	\[\tau \sqrt{\frac{N}{L^3}} \log N = \frac{\OUT}{\sqrt{NL^3}} \log N \le \frac{N^{\frac{3}{2}}\log N}{N^{\frac{1}{2}}(\frac{N}{p})^{\frac{3}{2}}} = \frac{p^\frac{3}{2}\log N}{N^{\frac{1}{2}}} \le 1,\]
 	where the last inequality follows from our assumption $N = \IN/3 \ge p^{3}\log^{2} \IN \ge p^{3}\log^2 N$. Then (\ref{eq:8}) can be simplified to
 	\begin{equation}
 	\label{eq:10}
 	\delta = c_3 \max\left\{\frac{1}{\tau} \sqrt{\frac{N}{L}} \log N, 1\right\}.
 	\end{equation}
 	Plugging (\ref{eq:9}) and (\ref{eq:10}) into $p\cdot J(L) =\Omega(\OUT) = \Omega(N\tau)$, we obtain
 	\begin{align*} L  =& \Omega\left(\frac{N}{\delta^{\frac{2}{3}} p^{\frac{2}{3}}}\right) = \Omega\left( \min \left \{ \left ({\tau \over \log N} \cdot \sqrt{\frac{L}{N}}  \right)^{\frac{2}{3}},  1 \right \} \cdot \frac{N}{p^{2/3}}\right). 
 	\end{align*}
 	Finally, after plugging in $\tau = \OUT / N$ and rearranging, we obtain
 	\[ L = \Omega\left(\min \left \{\frac{\OUT}{p  \log N}, \frac{\IN}{p^{2/3}} \right\}\right). \]
 \end{proof}
 
\noindent {\bf Remark.} Our lower bound has the following consequences:

\begin{enumerate}
	\item When $\OUT \ge \IN \cdot p^{1/3}$, the lower bound becomes $\tilde{\Omega} ({\IN \over p^{2/3}} )$, which means that the worst-case optimal algorithm of \cite{koutris16:_worst} is actually also output-optimal in this parameter range.  Finding $\tilde{\Omega}(\IN \cdot p^{1/3})$ triangles is as difficult as finding $\Theta(\IN^{3/2})$ triangles.
	\item When $\IN \le \OUT \le  \IN \cdot p^{1/3}$, the lower bound becomes $\tilde{\Omega}({\OUT \over p})$ while we do not have a matching upper bound yet.  Nevertheless, this already exhibits a separation from acyclic joins, which can be done with load $O ({\sqrt{\IN \cdot \OUT} \over p} )$.  The gap is at  least $\tilde{\Omega} (\sqrt{\OUT \over \IN} )$.
\end{enumerate}

%% file: appendix.tex
\appendix

\section{Proof of Lemma~\ref{lem:non-hierarchical}}
\label{sec:proof-lemma-path3}

\begin{proof}
  Direction ($\Leftarrow$): In an acyclic join $\Q = (\V, \E)$, a minimal path of length 3 is a sequence of 4 vertices $(x_1, x_2, x_3, x_4)$, such that $\{x_1, x_2\} \subseteq e_1, \{x_2, x_3\} \subseteq e_2, \{x_3, x_4\} \subseteq e_3$ and there exists no edge $e \in \E$ with $\{x_1, x_3\} \subseteq e$, or $\{x_1, x_4\} \subseteq e$, or $\{x_2, x_4\} \subseteq e$.  This already testifies that $\Q$ is not hierarchical.  To show that it is not r-hierarchical, consider the process of repeatedly applying the reduce procedure to $\Q$.  If any of $\{e_1, e_2, e_3\}$ is removed in the process, say $e_1$, there must exist an edge $e'_1$ such that $e_1 \subseteq e'_1, x_3 \notin e'_1, x_4 \notin e'_1$. The same applies for $e_2$ and $e_3$. Thus we can always find three edges $e'_1, e'_2, e'_3$ such that $e'_2 \in \E_{x_2} \cap \E_{x_3}, e'_1 \in \E_{x_2} - \E_{x_3}, e'_3 \in \E_{x_3} - \E_{x_2}$ after applying the reduce procedure, so this query is not r-hierarchical.
	
  \medskip Direction ($\Rightarrow$): The proof is constructive.  We will show below how to find a minimal path of length 3 in any acyclic but non-r-hierarchical join. We first apply the reduce procedure to $\Q$ such that no edge is contained in another. The rationale behind this is that a minimal path between two vertices $x,y \in \V$ of length $3$ in the reduced join is also a minimal path between $x,y$ of length 3 in the original join. Then we proceed in 3 steps: (We give an intuitive illustration of the results after each step, in Figure~\ref{fig:step}.)
	
	\medskip \noindent {\bf Step 1: } Find a subgraph defined by three distinct edges $\{e_1,e_2,e_3\}$ and four distinct vertices $\{x_1,x_2,x_3,x_4\}$, such that $x_1 \in e_1, x_1 \notin e_2 \cup e_3, x_2 \in e_1 \cap e_2, x_2 \notin e_3, x_3 \in e_2 \cap e_3, x_3 \notin e_1, x_4 \in e_3, x_4 \notin e_1 \cup e_2$.
	
	\medskip \noindent {\bf Step 2: } 
	Find a subgraph defined by three distinct edges $\{e_1,e_2,e_3\}$ and four distinct vertices $\{x_1,x_2,x_3,x_4\}$, such that $x_1 \in e_1, x_1 \notin e_2 \cup e_3, x_2 \in e_1 \cap e_2, x_2 \notin e_3, x_3 \in e_2 \cap e_3, x_3 \notin e_1, x_4 \in e_3, x_4 \notin e_1 \cup e_2$, and there exists no edge $e \in \E$ with $\{x_1, x_2, x_3\} \subseteq e$ or $\{x_2, x_3, x_4\} \subseteq e$.
	
	\medskip \noindent {\bf Step 3: } Find a minimal path of length $3$ between $x_1$ and $x_4$.
	\begin{figure}[h]
		\centering
		\includegraphics[scale=1.1]{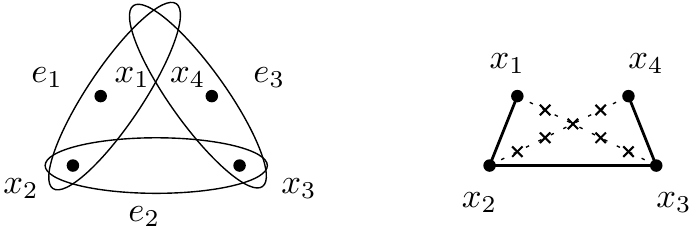}
		\caption{The left figure is the result of step 1 and the right figure is the result of step 2 $\&$ 3.}
		\label{fig:step}
	\end{figure}
	Our construction and its correctness proof is based on a basic property of acyclic join, as stated in Lemma~\ref{lem:acyclic-cycle}. 	With Lemma~\ref{lem:acyclic-cycle}, we are able to prove stronger results in Corollary~\ref{cor:acyclic-cycle-1} and Corollary~\ref{cor:acyclic-cycle-2}, which will be used as building blocks in proving Lemma~\ref{lem:non-hierarchical}. 
	
	\begin{lemma}
		\label{lem:acyclic-cycle}
		For three distinct edges $e_{xy}, e_{xz}, e_{yz} \in \E$, if $e_{xy} \cap e_{xz} - e_{yz} \neq \emptyset, e_{xy} \cap e_{yz} - e_{xz} \neq \emptyset, e_{xz} \cap e_{yz} - e_{xy} \neq \emptyset$, then there exists one edge $e \in \E$ such that $e_{xy} \cap e_{xz} \subseteq e, e_{xz} \cap e_{yz} \subseteq e, e_{xy} \cap e_{yz} \subseteq e$.
	\end{lemma}

	\noindent {\bf Proof of Lemma~\ref{lem:acyclic-cycle}: }  Consider attributes $x,y,z$ such that $x \in e_{xy} \cap e_{xz} - e_{yz}, y \in e_{xy} \cap e_{yz} - e_{xz}, z \in e_{xz} \cap e_{yz} - e_{xy}$. In the GYO reduction~\cite{abiteboul1995foundations}, we observe that (1) Any of $x,y,z$ won't be removed as an unique attribute before any edge of $e_{xy}, e_{xz}, e_{yz}$ is removed; (2) Any of $e_{xy}, e_{xz}, e_{yz}$ won't be removed as an empty edge before any of $x,y,z$ is removed. So it is always feasible to identify one edge $e \in \E$ such that $e_{xy} \cap e_{xz} - e_{yz} \subseteq e, e_{xy} \cap e_{yz} - e_{xz} \subseteq e, e_{xz} \cap e_{yz} - e_{xy} \subseteq e$. Moreover, any attribute in $e_{xy} \cap e_{xz} \cap e_{yz}$ if exists won't be removed as an unique attribute before any edge of $e_{xy}, e_{xz}, e_{yz}$ is removed. Thus we come to the conclusion in Lemma~\ref{lem:acyclic-cycle}.

	\begin{corollary}
	\label{cor:acyclic-cycle-1}
		For two distinct edges $e_{xy}, e_{xz} \in \E$ and a subset of edges $\E_{yz} \subseteq \E - \{e_{xy}, e_{xz}\}$, if $e_{xy} \cap e_{xz} - e_{yz} \neq \emptyset, e_{xy} \cap e_{yz} - e_{xz}$ for each $e_{yz} \in \E_{yz}$ and $\left(\bigcap_{e_{yz} \in \E_{yz}} (e_{xz} \cap e_{yz}) \right) - e_{xy} \neq \emptyset$, then there exists one edge $e \in \E$ such that $e_{xy} \cap e_{xz} \subseteq e, \bigcup_{e_{yz} \in \E_{yz}} (e_{xy} \cap e_{yz}) \subseteq e, \bigcap_{e_{yz} \in \E_{yz}} (e_{xz} \cap e_{yz}) \subseteq e$. 
	\end{corollary}
	
	\noindent {\bf Proof of Corollary~\ref{cor:acyclic-cycle-1}: } For simplicity, rename edges in $\E_{yz}$ as $e_1, e_2, \cdots, e_k$. We prove it by induction. The base case when $k = 1$ is precisely characterized and solved by Lemma~\ref{lem:acyclic-cycle}. We hold the hypothesis that there exists one edge $e \in \E$ such that 
	\[e \supseteq (e_{xy} \cap e_{xz})\cup \left(\bigcup_{i \in \{1,\cdots, k-1\}} (e_{xy} \cap e_i) \right) \cup \left(\bigcap_{i \in \{1,\cdots, k-1\}} (e_{xz} \cap e_i) \right).\]
	Moreover, if $e_{xy} \cap e_k \subseteq e$, edge $e$ is exactly the one characterized by Corollary~\ref{cor:acyclic-cycle-1} and we are done. Otherwise, $(e_{xy} \cap e_k) - e \neq \emptyset$. 
	
	We observe that $\left(\bigcap_{i \in \{1,\cdots, k\}} (e_{xz} \cap e_i) \right) \subseteq e \cap e_k$, so there is $(e \cap e_k) - e_{xy} \neq \emptyset$. If $e_{xy} \cap e - e_k = \emptyset$, there is $e_{xy} \cap e \subseteq e_k$. So far we have following observations on $e_k$ that (1) $e_k  \supseteq e_{xy} \cap e \supseteq (e_{xy} \cap e_{xz})\cup \left(\bigcup_{i \in \{1,\cdots, k-1\}} (e_{xy} \cap e_i) \right)$; (2) $e_k  \supseteq e_{xy} \cap e_k$; (3) $e_k  \supseteq e_{xz} \cap e_k \supseteq \bigcap_{i \in \{1,\cdots, k\}} (e_{xz} \cap e_i)$, 
	or equivalently,
	\[e_k \supseteq (e_{xy} \cap e_{xz})\cup \left(\bigcup_{i \in \{1,\cdots, k\}} (e_{xy} \cap e_i) \right) \cup \left(\bigcap_{i \in \{1,\cdots, k\}} (e_{xz} \cap e_i) \right).\]
	Thus edge $e_k$ is exactly the one characterized by Corollary~\ref{cor:acyclic-cycle-1}, and we are done. Otherwise, $e_{xy} \cap e - e_k \neq \emptyset$. Implied by Lemma~\ref{lem:acyclic-cycle}, there exists an edge $e' \in \E$ such that $e_{xy} \cap e_k \subseteq e', e_{xy} \cap e  \subseteq e', e_k \cap e \subseteq e'$. More precisely, (1) $e' \supseteq e_{xy} \cap e \supseteq (e_{xy} \cap e_{xz})\cup \left(\bigcup_{i \in \{1,\cdots, k-1\}} (e_{xy} \cap e_i) \right)$; (2) $e' \supseteq e_{xy} \cap e_k$;
	(3) $e'\supseteq e_k \cap e \supseteq \bigcap_{i \in \{1,\cdots, k\}} (e_{xz} \cap e_i)$. Or equivalently,
	\[e' \supseteq (e_{xy} \cap e_{xz})\cup \left(\bigcup_{i \in \{1,\cdots, k\}} (e_{xy} \cap e_i) \right) \cup \left(\bigcap_{i \in \{1,\cdots, k\}} (e_{xz} \cap e_i) \right),\]
	thus edge $e'$ is exactly the one characterized by Corollary~\ref{cor:acyclic-cycle-1}, and we are done.

	\begin{corollary}
	\label{cor:acyclic-cycle-2}
		For a set of distinct vertices $x, y_1, y_2, \cdots, y_k$, if there exists one edge $e_0 \in \E$ such that $x \notin e_0, \{y_1, y_2, \cdots, y_k\} \\\subseteq e_0$, and there exists one edge $e_i \in \E$ such that $\{x, y_i\} \subseteq e_i$ for each $i \in \{1,2,\cdots,k\}$, then there exists one edge $e' \in \E$ such that $\{x, y_1, y_2, \cdots, y_k\} \subseteq e'$.
	\end{corollary}
	
	\noindent {\bf Proof of Corollary~\ref{cor:acyclic-cycle-2}: } We prove it by induction. The base case when $k = 1$ is trivial. We hold the hypothesis that there exists one edge $e \in \E$ such that $\{x, y_1, y_2, \cdots, y_{k-1}\} 
	\subseteq e$. 
	
	If $y_k \in e$, edge $e$ is exactly the one characterized by Corollary~\ref{cor:acyclic-cycle-2} and we are done. Moreover, if $\{y_1, y_2, \cdots, y_{k-1}\} \subseteq e_k$, edge $e_k$ is exactly the one characterized by Corollary~\ref{cor:acyclic-cycle-2} and we are done. Otherwise, $y_k \notin e$ and $\{e_1,e_2, \cdots, e_{k-1}\} - e_k \neq \emptyset$. Note that $y_k \in e_k \cap e_0 - e$, $x \in e_k \cap e - e_0$, and $e \cap e_0 - e_k \neq \emptyset$.
	Implied by Lemma~\ref{lem:acyclic-cycle}, there exists one edge $e' \in \E$ such that $e \cap e_k \subseteq e'$, $e_k \cap e_0 \subseteq e'$ and $e_0 \cap e \subseteq e'$. More precisely, $\{e_1, e_2, \cdots, e_{k-1}\} \subseteq e \cap e_0 \subseteq e'$, $\{x\} \subseteq e \cap e_k \subseteq e'$, and $\{e_k\} \subseteq e_k \cap e_0 \subseteq e'$, thus $\{x,y_1, y_2, \cdots, y_k\} \subseteq e'$.

	\medskip \noindent {\bf Proof of step 1: }
	
	If an acyclic join  is not r-hierarchical, then there exist two attributes $x,y$ such that $\E_x \cap \E_y \neq \emptyset, \E_x - \E_y \neq \emptyset, \E_y - \E_x \neq \emptyset$. Consider $e_{xy} \in \E_x \cap \E_y$, $e_x \in \E_x - \E_y$ and $e_y \in \E_y - \E_x$. It suffices to show that $e_x - e_{xy} - e_y \neq \emptyset$ and $e_y - e_{xy} - e_x \neq \emptyset$ by the constraint. First $e_x - e_{xy}$ is not empty otherwise $e_x \subseteq e_{xy}$ contradicting our assumption. The same applies for $e_y - e_{xy} \neq \emptyset$. If $e_x - e_{xy} - e_y = \emptyset$, each attribute appearing in $e_x - e_{xy}$ also appears in $e_y$. In this way, we can identify three distinct attributes $x,y,z$ such that $x \in e_{x} \cap e_{xy} - e_y$, $y \in e_y \cap e_{xy} - e_x$, $e_{x} \cap e_{y} - e_{xy}$, which form a cycle. Thus there exists an edge $e_{xyz} \in \E$ such that $e_x \cap e_{xy}  \subseteq e_{xyz}, e_y \cap e_{xy} \subseteq e_{xyz}, e_x \cap e_y  \subseteq e_{xyz}$ implied by Lemma~\ref{lem:acyclic-cycle}. Note that $e_x - e_{xy} - e_y = \emptyset$ implies that $e_x$ can be rewritten as $(e_x \cap e_{xy}) \cup (e_x \cap e_y)$. In this way, $e_x \subseteq e_{xyz}$ contradicting our assumption. So we have $e_x - e_{xy} -e_y  \neq \emptyset$, and the same applies for $e_y - e_{xy} - e_x  \neq \emptyset$.
	
	\medskip \noindent {\bf Proof of step 2: }
	
	Assume we already have a subgraph defined by edges $\{e_1,e_2,e_3\}$ and vertices $\{x_1,x_2,x_3,x_4\}$, such that $x_1 \in e_1, x_1 \notin e_2 \cup e_3, x_2 \in e_1 \cap e_2, x_2 \notin e_3, x_3 \in e_2 \cap e_3, x_3 \notin e_1, x_4 \in e_3, x_4 \notin e_1 \cup e_2$. If there exists no edge $e \in \E$ such that $\{x_1, x_2, x_3\} \subseteq e$ or $\{x_2, x_3, x_4\} \subseteq e$, we are done. Otherwise, we need to show how to find $x_1', x_4'$ satisfying our condition to replace $x_1, x_4$. Note that the replacement of $x_1$ and that of $x_4$ are independent, as well as their correctness arguments.
	
	In the following, we will tackle the situation where there exists an edge $e \in \E$ such that $\{x_1,x_2,x_3\} \subseteq e$. The situation where there exists an edge $e \in \E$ such that $\{x_2, x_3, x_4\} \subseteq e$ is symmetric and can be tackled similarly. 
	
	Define the attribute set $S = \{x \in e_1: \exists e \in \E, \{x_2, x_3, x\} \subseteq e\}$. If $e_1 - e_2 - e_3 - S \neq \emptyset$, then we just replace $x_1$ by any attribute in $e_1 - e_2 - e_3 - S$. Otherwise, $e_1 - e_2 - e_3 - S = \emptyset$, which implies that $e_1$ can be rewritten as $(e_1 \cap S) \cup (e_1 \cap e_2) \cup (e_1 \cap e_3) $. We will prove by contradiction that this case won't happen in the reduced join. Define the edge set $\E_S = \{e \in \E: \exists x \in S, \{x_2, x_3, x\} \subseteq e\}$. Note that if $x \notin S$, then $x \notin e$ for each $e \in \E_S$. We distinguish following four cases. We give an intuitive illustration of the contradiction in each case, in Figure~\ref{fig:step2}. The same technique we adopt is to identify an edge $e \in \E$ such that $e_1 \neq e$ and $e_1 \subseteq e$, coming to a contradiction in a reduced join.
	\begin{figure}[h]
		\centering
		\includegraphics[scale=1.1]{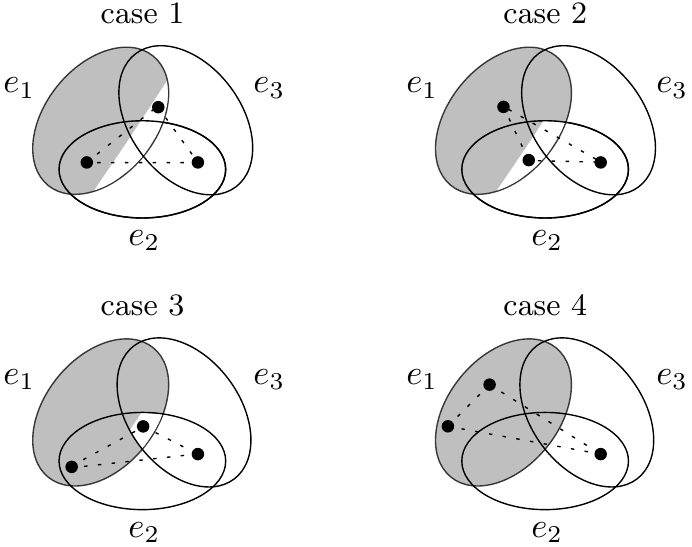}
		\caption{Illustration of four cases in step 2. The shade area represents the attribute set $S$.}
		\label{fig:step2}
	\end{figure}
	
	\paragraph{Case 1: $\bm{e_1 \cap e_3 - e_2 - S \neq \emptyset}$} 
	
	Consider an arbitrary attribute $x \in e_1 \cap e_3 - e_2 - S$. Denote $\E'_S = \{e_2\} \cup \E_S$. Note that $x_2 \in e_1 \cap e - e_3$, $x_3 \in e \cap e_3 - e_1$, and $x \in e_1 \cap e_3 - e$ for each $e \in \E'_S$. Implied by Corollary~\ref{cor:acyclic-cycle-1}, there exists an edge $e' \in \E$ such that $e_1 \cap e_3 \subseteq e'$, $x_3 \in e'$ and $e_1 \cap e \subseteq e'$ for each $e \in \E'_S$. This also implies $e_1 \cap S \subseteq e'$, $e_1 \cap e_2 \subseteq e'$, and $e_1 \neq e'$. Thus, $e_1 \subseteq e'$ contradicting our assumption. 
	
	\paragraph{Case 2: $\bm{e_1 \cap e_3 - e_2 - S = \emptyset}$ and $\bm{e_1 \cap e_2 - e_3 - S \neq \emptyset}$} 
	
	Consider an arbitrary attribute $x \in e_1 \cap e_2 - e_3 - S$. Denote $S' = S - e_2$, where $S' \neq \emptyset$ since $x_1 \in S - e_2$. Note that $x \in e_1 \cap e_2 - e$, $e \cap e_1 - e_2 \neq \emptyset$, and $x_3 \in e_2 \cap e - e_1$ for each $e \in \E_S'$. Implied by Corollary~\ref{cor:acyclic-cycle-1}, there exists an edge $e' \in \E$ such that $e_1 \cap e_2 \subseteq e'$, $x_3 \in e'$ and $e_1 \cap e \subseteq e'$ for each $e \in \E_S'$. This also implies $e_1 \cap S' \subseteq e'$ and $e_1 \neq e'$. Thus $(e_1 \cap S) \cup (e_1 \cap e_2) \subseteq e$. We already have $e_1 \cap e_3 - e_2 - S = \emptyset$ in this case. Thus, $e_1 \subseteq e$ contradicting our assumption.
	
	\paragraph{Case 3: $\bm{e_1 \cap e_3 - e_2 - S = \emptyset}$, $\bm{e_1 \cap e_2 - e_3 - S = \emptyset}$, and $\bm{e_1 \cap e_2 \cap e_3 - S \neq \emptyset}$} 
	
	Consider an arbitrary attribute $x \in e_1 \cap e_2 \cap e_3 - S$. Note that $x_2 \in e_1 \cap e - e_3$, $x_3 \in e_3 \cap e - e_1$, and $x \in e_1 \cap e_3  - e$ for each $e \in \E_S$. Implied by Corollary~\ref{cor:acyclic-cycle-1}, there exists an edge $e' \in \E$ such that $e_1 \cap e_3 \subseteq e'$, $x_3 \in e'$ and $e_1 \cap e \subseteq e'$ for each $e \in \E_S$. This also implies $e_1 \cap S \subseteq e'$ and $e_1 \neq e'$. We already have $e_1 \cap e_2 - e_3 - S = \emptyset$ in this case. Thus, $e_1 \subseteq e'$ contradicting our assumption. 
	
	\paragraph{Case 4: $\bm{e_1 \cap e_3 - e_2 - S = \emptyset}$, $\bm{e_1 \cap e_2 - e_3 - S = \emptyset}$, and $\bm{e_1 \cap e_2 \cap e_3 - S = \emptyset}$} 
	
	Under this circumstances, $e_1 \subseteq S$. Implied by the fact that $S \subseteq e_1$, we have $e_1 = S$. For attributes $x_3$ and all attributes in $S$, there is $S \subseteq e_1$, and for each $x \in S$ there exists one edge $e_x \in \E_S$ such that $\{x, x_3\} \subseteq e_x$. Implied by Corollary~\ref{cor:acyclic-cycle-2}, there exists one edge $e' \in \E'$ such that $x_3 \in e'$ and $S \subseteq e'$. Thus, $e_1 \neq e'$ and $e_1 \subseteq e'$, contradicting our assumption.

	\medskip Combining these four cases proves the step 2.

	\medskip \noindent {\bf Proof of step 3: } Consider the subgraph found in the last step. By the definition of minimal path, it suffices to show that there exists no edge $e' \in \E$ such that $\{x_1,x_3\} \subseteq e'$, or $\{x_1,x_4\} \subseteq e'$, or $\{x_2, x_4\} \subseteq e'$. By contradiction, assume there is an $e'$ where $\{x_1,x_3\} \subseteq e'$. Implied by the contraints of this subgragh, $x_2 \notin e'$ and $x_4 \notin e'$. Attributes $x_1, x_2, x_3$ form a cycle on edges $e_1, e_2, e'$, then there must exist an edge containing all of $\{x_1, x_2, x_3\}$ contradicting the constraints. The similar argument applies for $\{x_1, x_4\} \subseteq e'$ and $\{x_2, x_4\} \subseteq e'$.
\end{proof}